\documentclass[11pt]{article}
\usepackage{amsmath,amssymb,amsthm,cite}
\usepackage{fullpage}
\usepackage{color}
\usepackage{xcolor}
\usepackage[normalem]{ulem}
\usepackage{graphicx}
\usepackage{todonotes}
\usepackage[framemethod=tikz]{mdframed}
\usepackage{algorithm}
\usepackage[noend]{algpseudocode}
\usepackage[pagebackref=false]{hyperref}
\hypersetup{
    unicode=false,          
    colorlinks=true,        
    linkcolor=red,          
    citecolor=green,        
    filecolor=magenta,      
    urlcolor=cyan           
}
\usepackage[capitalize]{cleveref}

\newcommand{\newblue}[1]{{\color{blue}#1}}
\newcommand{\eps}{\epsilon}

\newcommand{\ceil}[1]{\left\lceil #1 \right\rceil}

\newtheorem{theorem}{Theorem}[section]

\newtheorem{lemma}[theorem]{Lemma}

\newtheorem{corollary}[theorem]{Corollary}
\newtheorem{remark}[theorem]{Remark}
\newtheorem{hypothesis}[theorem]{Hypothesis}
\newtheorem{rules}[theorem]{Rules}

\algnewcommand\algorithmicswitch{\textbf{switch}}
\algnewcommand\algorithmiccase{\textbf{case}}

\algdef{SE}[SWITCH]{Switch}{EndSwitch}[1]{\algorithmicswitch\ #1\ \algorithmicdo}{\algorithmicend\ \algorithmicswitch}%
\algdef{SE}[CASE]{Case}{EndCase}[1]{\algorithmiccase\ #1}{\algorithmicend\ \algorithmiccase}%
\algtext*{EndSwitch}%
\algtext*{EndCase}%

\algnewcommand\algorithmicwithprob{\textbf{with probability}}
\algnewcommand\algorithmicotherwise{\textbf{otherwise}}

\algdef{SE}[WithProb]{WithProb}{EndWithProb}[1]{\algorithmicwithprob\ #1\ \algorithmicdo}{\algorithmicend\ \algorithmicwithprob}%
\algdef{Ce}[Otherwise]{WithProb}{Otherwise}{EndWithProb}
  {\algorithmicotherwise}%
\algtext*{EndWithProb}%
\algtext*{EndOtherwise}%

\begin{document}

\date{}
\title{Optimal Error Rates for Interactive Coding I: \\Adaptivity and Other Settings\\{\small \color{blue} (Corrected version)}}

%

\author{
  Mohsen Ghaffari\\
  \small MIT\\
  {\small ghaffari@mit.edu}
  \and
  Bernhard Haeupler\\
  \small Microsoft Research \\
  {\small haeupler@alum.mit.edu}
	\and
	Madhu Sudan\\
  \small Microsoft Research \\
  {\small madhu@mit.edu}
}

\maketitle
\setcounter{page}{0}
\thispagestyle{empty}

\begin{abstract}
We consider the task of interactive communication in the presence of adversarial errors and present tight bounds on the tolerable error-rates in a number of different settings. 

\smallskip

Most significantly, we explore {\em adaptive} interactive communication where the communicating parties decide who should speak next based on the history of the interaction. In particular, this allows this decision to depend on estimates of the amount of errors that have occurred so far. Braverman and Rao [STOC'11] show that non-adaptively one can code for any constant error rate below $1/4$ but not more. They asked whether this bound could be improved using adaptivity. We answer this open question in the affirmative (with a slightly different collection of resources): Our adaptive coding scheme tolerates any error rate below $2/7$ and we show that tolerating a higher error rate \newblue{than $1/3$} is impossible. We also show that in the setting of Franklin et al. [CRYPTO'13], where parties share randomness not known to the adversary, adaptivity increases the tolerable error rate from $1/2$ to $2/3$ {\color{blue} (and show a matching impossibility result for tolerating a higher error rate than $2/3$)}. For list-decodable interactive communications, where each party outputs a constant size list of possible outcomes, the tight tolerable error rate is $1/2$. 

\smallskip

Our negative results hold even if the communication and computation are unbounded, whereas for our positive results communication and computation are {\em polynomially} bounded. Most prior work considered coding schemes with {\em linear} amount of communication, while allowing unbounded computations. We argue that studying tolerable error rates in this relaxed context helps to identify a setting's intrinsic optimal error rate. We set forward a strong working hypothesis which stipulates that for any setting the maximum tolerable error rate is independent of many computational and communication complexity measures. We believe this hypothesis to be a powerful guideline for the design of simple, natural, and efficient coding schemes and for understanding the (im)possibilities of coding for interactive communications.

\end{abstract}

\newpage
{\color{violet}
\begin{center}
{\bf \Large Erratum to "Optimal Error Rates for Interactive Coding I: Adaptivity and Other Settings"~\cite{stoc2014optimal}} \\
~\\
{Mohsen Ghaffari, Bernhard Haeupler, and  Madhu Sudan} 
\end{center} 

This erratum reports an error in one of the claims of the paper "Optimal Error Rates for Interactive Coding: Adaptivity and Other Settings" from STOC 2014~\cite{stoc2014optimal}. 

\smallskip

The aforementioned paper proposed the first formal model for \emph{adaptive} interactive coding schemes. The paper then gave results for how much utilizing adaptivity can improve the maximal tolerable error rate, under which interactive coding schemes can operate reliably, in different settings. In particular, the paper proved matching upper and lower bounds for the maximal tolerable rate of erasures that can be corrected in interactive communications,  showing that adaptivity increases the tolerable error rate from $1/2$ to $2/3$. The paper was also the first to introduce and study list-decodable interactive communications, where each party outputs a constant size list of possible protocol outcomes. The paper proved that the tolerable error rate for list-decodable interactive coding is a tight $1/2$, with or without adaptivity. For regular errors it gave an adaptive coding scheme tolerating an rate of up to $2/7$. Lastly, the paper introduced Adaptive Message Exchange as the arguably simplest interactive protocol and claimed a matching lower bound of $2/7$ on the maximal tolerable error rate, even when restricting adaptive interactive coding to this Adaptive Message Exchange Problem. It is this claimed impossibility result for Adaptive Message Exchange with error rate $2/7$ (stated as Theorem 4.3) and its proof that are incorrect. A weaker impossibility result of $1/3$ (stated as Theorem 4.2 and proven in Appendix D) for the exact same setting is correct as stated. This correct $1/3$ impossibility result was given in the original paper as a "warm-up" to the incorrect $2/7$ claim. 

\smallskip

Next, we give a brief description of where in the claimed "proof" of Theorem 4.3 the error occurs (credit for identifying this error so cleanly goes to Ran Gelles): The proof implicitly assumes that for the $N$ round protocol to be simulated there exists a pair of inputs such that the time $x_A$ at which Alice has listened (alone) for $2N/7$ of the time and the time $x_B$ at which Bob has listened (alone) for $2N/7$ of the time are (essentially) the same. While such a symmetry might seem very natural for protocols for the (symmetric) Message Exchange Problem it cannot be assumed without loss of generality, as is done implicitly in the incorrect "proof". Concretely, the "proof" breaks in the case where $x_A, x_B > 4N/7$. In this case the adversary designed by the "proof" starts corrupting both parties after they have listened for $2N/7$ rounds (alone) respectively. In asymmetric protocols, the party who initially listens (alone) more frequently (despite no corruptions having happened) reaches their $2N/7$ threshold first and the adversary will start corrupting the messages received by this party, say Alice. This might prompt Alice to then report these inconsistencies back to Bob, before the $2N/7$ threshold of Bob has been reached. Once Bob hears these inconsistency reports without the adversary being allowed to corrupt them (yet) the indistinguishably argument of the "proof" breaks.

\smallskip

The above error was reported to us by Ran Gelles in 2015 and acknowledged on the second author's website. We are tremendously thankful for the significant time and effort Ran Gelles spent on clearly identifying the error and discussing the problem with us. We are even more thankful for the extremely generous, patient, courteous, and polite manner in which Ran Gelles helped us with this bug. 

\smallskip

Since STOC 2014 the model for adaptive interactive coding schemes proposed in~\cite{stoc2014optimal} and their tolerable error rates have been studied further, both for the simple Adaptive Message Exchange problem as well as for the fully general setting of achieving (constant-rate) adaptive interactive coding schemes for any interactive protocol. Specifically, the state-of-the-art is currently as follows:
\begin{itemize}
        \item The second part ``{\bf Optimal Error Rates for Interactive Coding II: Efficiency and List decoding}'', by the first and second author~\cite{focs2014optimal} build on the coding schemes from~\cite{stoc2014optimal} and transformed all of them to have constant-communication rate and polynomial computational efficiency. This includes the adaptive simulation for regular errors tolerating an error rate of up to $2/7$ given in~ \cite{stoc2014optimal} in improves . 
        \item For the Adaptive Message Exchange Protocol novel and tightly matching upper and lower bounds are proven in  ``{\bf Optimal Error Resilience of Adaptive Message Exchange}'', by Klim Efremenko, Gillat Kol, and Raghuvansh R. Saxena in~\cite{efremenko2021optimal}. They show that the Message Exchange protocol can be robustly simulated by an adaptive coding scheme with up to $5/16$ fraction of errors and that this is tight, i.e., Adaptive Message Exchange is impossible with more than $5/16$ fraction of errors. Note that $5/16$ lies strictly between the $2/7$ and $1/3$ error rate upper and lower bounds proven in~\cite{stoc2014optimal}.

\item In ``{\bf Interactive error resilience beyond $2/7$}'', by Klim Efremenko, Gillat Kol, and Raghuvansh R. Saxena~\cite{efremenko2020interactive} it is shown that fully general constant-rate interactive coding is possible for any protocol with up to $7/24 > 2/7$ fraction of errors using adaptivity. 

\item Overall, this puts the maximal tolerable error rate for (constant-rate adaptive) interactive coding somewhere in $[7/24,5/16]$. Whether the adaptive coding scheme for Message Exchange, which tolerates an error rate of up to $5/16$, can be generalized to general interactive protocols, or whether a stronger-than-$5/16$ impossibility result for general interactive coding can be proven remain interesting open questions. 

\end{itemize}

In what follows we include a corrected version of the paper. Text that has been \sout{struckthrough} and text in {\color{red} red} should be considered deleted (though it has been left in place to allow easier comparison with the previous paper). Text in {\color{blue} blue} has been added. 

}

\newpage
\section{Introduction}

``Interactive Coding'' or ``Coding for Interactive Communication'' studies the task of protecting an interaction between two parties in the presence of communication errors. This line of work was initated by Schulman~\cite{Schulman} who showed, surprisingly at the time, that protocols with $n$ rounds of communication 
can be protected against a (small) constant fraction of adversarial errors while incurring only a constant overhead in the total communication complexity. 


In a recent powerful result that revived this area, Braverman and Rao~\cite{BR11} explored the maximal rate of errors that could be tolerated in an interactive coding setting. They showed the existence of a protocol that handles a $1/4 - \epsilon$ error rate and gave a matching negative result under the assumption that the coding scheme is {\em non-adaptive} in deciding which player transmits (and which one listens) at any point of time. They left open the questions whether the $1/4$ error rate can be improved by allowing adaptivity (see  \cite[Open Problem 7]{BravermanAllerton} and \cite[Conclusion]{BR11}) or by reducing the decoding requirement to list decoding (see \cite[Open Problem 9]{BravermanAllerton} and \cite[Conclusion]{BR11}), that is, requiring each party only to give a small list of possible outcomes of which one has to be correct. 



In this work we answer both questions in the affirmative (in a somewhat different regime of computational and communication resources): We give a rate adaptive coding scheme that tolerates any error rate below $2/7$. We furthermore show a \sout{matching} impossibility result which strongly rules out any coding scheme achieving an error rate of \sout{$2/7$} {\color{blue} $1/3$}.

Moreover, we also consider the adaptive coding schemes in the setting of \cite{FGOS13} in which both parties share some randomness not known to the adversary. While non-adaptive coding schemes can tolerate any error rate below $1/2$ this bound increases to $2/3$ using adaptivity, which we show is also best possible. 

Lastly, we initiate the study of list decodable interactive communication. We show that allowing both parties to output a constant size list of possible outcomes allows non-adaptive coding schemes that are robust against any error rate below $1/2$, which again is best possible for in both the adaptive and non-adaptive setting.

All our coding schemes are deterministic and work with communication and computation being polynomially bounded in the length of the original protocol. We note that most previous works considered the more restrictive setting of linear amount of communication (often at the cost of exponential time computations). Interestingly, our matching negative results hold even if the communication and computation are unbounded. We show that this sharp threshold behavior extends to many other computational and communication complexity measures and is common to all settings of interactive communication studied in the literature. In fact, an important conceptual contribution of this paper is the formulation of a strong working hypothesis that stipulates that maximum tolerable error rates are invariable with changes in complexity and efficiency restrictions on the coding scheme. Throughout this paper this hypothesis lead us to consider the simplest setting for positive results and then expanding on the insights derived to get the more general positive results. We believe that in this way, the working hypothesis yields a powerful guideline for the design of simple and natural coding schemes as also the search for negative results. This has been already partially substantiated by subsequent results (see \cite{GH13} and \Cref{sec:appendix:furtherresultssupportingthehypothesis}).

\paragraph{Organization}
In what follows, 
we briefly introduce the interactive communication model more formally
in \Cref{sec:settingDefinitions}. We also introduce the model 
for adaptive interaction there. Then, in \Cref{sec:overview}, we explain our results as well as the underlying high-level ideas and techniques.
In \Cref{sec:exchange} we describe the simple Exchange problem
and give an adaptive protocol tolerating $2/7$-fraction of errors
in \Cref{sec:27rateadaptationXOR}. (Combined with \Cref{sec:overview}
this section introduces all the principal ideas of the paper. Rest
of the paper may be considered supplemental material.)
In the remainder of \mbox{\Cref{sec:exchange}}, we prove that an error-rate of \sout{$2/7$ is the best achievable} \newblue{$1/3$ is impossible to exceed} for the Exchange problem and thus also for the general case of interactive communication.
In \Cref{sec:simulators}, we give interactive coding schemes over
large alphabets tolerating $2/7$ error rate for general interactions. In
\Cref{sec:AlgLargeRounds} we then convert these to coding schemes
over constant size alphabets.
Finally, in \Cref{sec:sharedrand}, we give protocols tolerating an $2/3$ error rate
in the presence of shared randomness. The appendix contains some technical proofs, as well as some simple
impossibility results showing tightness of our protocols.

%
%
%

\section{New and Old Settings for Interactive Coding}\label{sec:settingDefinitions} 
In this section, we define the classical interactive coding setup as well as all new settings considered in this work, namely, list decoding, the shared randomness setting, and adaptive protocols.

We start with some standard terminology: 
%
 An $n$-round {\em interactive protocol} $\Pi$ between two players Alice and Bob is given by two functions $\Pi_A$ and $\Pi_B$. For each {\em round} of communication, these functions map (possibly probabilistically) the history of communication and the player's private input to a decision on whether to listen or transmit, and in the latter case also to a symbol of the {\em communication alphabet}. All protocols studied prior to this work are {\em non-adaptive}\footnote{Braverman and Rao~\cite{BR11} referred to protocols with this property as \emph{robust}.} in that the decision of a player to listen or transmit deterministically depends only on the round number, ensuring that exactly one party transmits in each round. In this case, the {\em channel} delivers the chosen symbol of the transmitting party to the listening party, unless the adversary interferes and alters the symbol arbitrarily. In the adversarial channel model with {\em error rate} $\rho$, the number of such errors is at most $\rho n$. The outcome of a protocol is defined to be the transcript of the interaction. 
%
%
%

A protocol $\Pi'$ is said to {\em robustly simulate} a protocol $\Pi$ for an error rate $\rho$ if the following holds: Given any inputs to $\Pi$, both parties can {\em uniquely decode} the transcript of an error free execution of $\Pi$ on these inputs from the transcript of any execution of $\Pi'$ in which at most a $\rho$ fraction of the transmissions were corrupted. This definition extends easily to {\em list decoding} by allowing both parties to produce a small (constant size) list of transcripts that is required to include the correct decoding, i.e., the transcript of $\Pi$. We note that the simulation $\Pi'$ typically uses a larger alphabet and a larger number of rounds. While our upper bounds are all deterministic, we strengthen the scope of our lower bounds by also considering \emph{randomized protocols} in which both parties have access to independent private randomness. We also consider the setting of \cite{FGOS13} in which both parties have access to \emph{shared randomness}. In both cases we assume that the adversary does not know the randomness and we say a randomized protocol robustly simulates a protocol $\Pi$ with \emph{failure probability} $p$ if, for any input and any adversary, the probability that both parties correctly (list) decode is at least $1 - p$.
%
 
%
We now present the notion of an {\em adaptive} protocol. It turns out that defining a formal model for adaptivity leads to several subtle issues. We define the model first and discuss these issues later.

In an {\em adaptive} protocol, the communicating players are allowed to base their decision on whether to transmit or listen (probabilistically) on the communication history. In particular, this allows players to base their decision on estimates of the amount of errors that have happened so far (see \Cref{sec:whyadaptivityhelps} for why this kind of adaptivity is a natural and useful approach). This can lead to rounds in which both parties transmit or listen simultaneously. In the first case no symbols are delivered while in the latter case the symbols received by the two listening parties are chosen by the adversary, without it being counted as an error. 

\paragraph{Discussion on the adaptivity model.}
It was shown in \cite{BR11} that protocols which under no circumstances have both parties transmit or listen simultaneously are necessarily non-adaptive. Any model for adaptivity must therefore allow parties to simultaneously transmit or listen and specify what happens in either case. Doing this and also deciding on how to measure the amount of communication and the number of errors leads to several subtle issues.

While it seems pessimistic to assume that the symbols received by two simultaneously listening parties are determined by the adversary this is a crucial assumption. If, e.g., a listening party could find out without doubt if the other party transmitted or listened by receiving silence in the latter case then uncorrupted communication could be arranged by simply using the listen/transmit state as an incorruptible one-bit communication channel. More subtle points arise when considering how to define the quantity of communication on which the adversaries budget of corruptions is based. The number of transmissions performed by the communicating parties, for example, seems like a good choice. This however would make the adversaries budget a variable (possibly probabilistic) quantity that, even worse, non-trivially depends on when and how this budgets is spent. It would furthermore allow parties to time-code, that is, encode a large number (even an encoding of all answers to all possible queries) in the time between two transmissions. While time-coding strategies do not seem to lead to very efficient algorithms they would prevent strong lower bounds which show that even over an unbounded number of rounds no meaningful communication is possible (see, e.g., \Cref{thm:27lowerboundoverview} which proves exactly this for an error rate of \sout{$2/7$} \newblue{$1/3$}).

Our model avoids all these complications. For non-adaptive protocols that perfectly coordinate a designated sender and receiver in each round our model matches the standard setting. For the necessary case that adaptive parties fail to coordinate our model prevents any signaling or time-sharing capabilities and in fact precludes any information exchange. This matches the intuition that in a conversation no advantage can be derived from both parties speaking or listening at the same time. It also guarantees that the product between the number of rounds and the bit-size of the communication alphabet is a clean and tight information theoretic upper bound on the total amount of information or entropy that can be exchanged (in either direction) between the two parties. This makes the number of rounds the perfect quantity to base the adversaries budget on. All this makes our model, in hindsight, the arguably cleanest and most natural extension of the standard setting to adaptive protocols (see also \Cref{sec:appendix:naturalmodel} for a natural interpretation as a wireless channel model with bounded unpredictable noise). The strongest argument for our model however is the fact that it allows to prove both strong and natural positive \emph{and} negative results, implying that our model does not overly restrict or empower the protocols or the adversary.



\section{Overview} \label{sec:overview}
In this section we state our results and explain the high level ideas and insights behind them. 

\subsection{Adaptivity}

A major contribution of this paper is to show that adaptive protocols can tolerate more than the $1/4$ error rate of the non-adaptive setting:

\begin{theorem}\label{thm:27upperboundoverview}
Suppose $\Pi$ is an $n$-round protocol over a constant bit-size alphabet. For any $\eps >0$, there is a deterministic computationally efficient protocol $\Pi'$ that robustly simulates $\Pi$ for an error rate of $2/7 - \eps$ using $O(n^2)$ rounds and an $O(1)$-bit size alphabet.
\end{theorem}

The proof is presented in \Cref{sec:AlgLargeRounds}. Furthermore, in \mbox{\Cref{subsec:xor13impossible}}, we show a \sout{matching} impossibility result  which even applies to the arguably simplest interactive protocol, namely, the \emph{Exchange Protocol}. In the Exchange Protocol each party simply gets an input---simply a one bit input in our impossibility results---and sends this input to the other party.

\begin{theorem}\label{thm:27lowerboundoverview}
There is no (deterministic or randomized) protocol that robustly simulates the Exchange Protocol for an error rate of \sout{$2/7$} \newblue{$1/3$} with an $o(1)$ failure probability even when allowing computationally unbounded protocols that use an arbitrarily large number of rounds and an unbounded alphabet.
\end{theorem}

\paragraph{Why Adaptivity is Natural and Helpful}\label{sec:whyadaptivityhelps}

Next, we explain why it should not be surprising that adaptivity leads to a higher robustness. We also give some insights for why the $2/7$ error rate is \sout{the} {\color{blue} a} natural tolerable error rate for adaptive protocols. 

It is helpful to first understand why the $1/4$ error rate was thought of as a natural barrier. The intuitive argument, presented in \cite{BR11}, for why one should not be able to cope with an error rate of $1/4$ is as follows: During any $N$ round interaction one of the parties, w.l.o.g. Alice, is the designated sender for at most half of the rounds. With an error rate of $1/4$ the adversary can corrupt half of the symbols Alice sends out. This makes it impossible for Alice to (reliably) deliver even a single input bit $x$ because the adversary can always make the first half of her transmissions consistent with $x=0$ and the second half with $x=1$ without Bob being able to know which of the two is real and which one is corrupted. 

While this intuition is quite convincing at the first glance, it silently assumes that it is a priori clear which of the two parties transmits less often. This in turn essentially only holds for non-adaptive protocols for which the above argument can also be made into a formal negative result\cite[Claim 10]{BR11}. On the other hand, we show that assuming this a priori knowledge is not just a minor technical assumption but indeed a highly nontrivial restriction which is violated in many natural settings of interaction. For example, imagine a telephone conversation on a connection that is bad/noisy in one direction. One person, say Alice, clearly understands Bob while whatever Alice says contains so much noise that Bob has a hard time understanding it. In a non-adaptive conversation, Bob would continue to talk half of the time (even though he has nothing to say given the lack of understandable responses from Alice) while Alice continues to talk little enough that she can be completely out-noised. This is of course not how it would go in real life. There, Bob would listen more in trying to understand Alice and by doing this give Alice the opportunity to talk more. Of course, as soon as this happens, the adversary cannot completely out-noise Alice anymore and the conversation will be able to progress. In fact, similar \emph{dynamic rate adaptation} mechanisms that adapt the bitrate of a senders to channel conditions and the communication needs of other senders are common in many systems, one prominent example being IEEE 802.11 wireless networks. 

Even if one is convinced that adaptive algorithms should be able to beat the $1/4$ error rate, it is less clear at this point what the maximum tolerable error rate should be. In particular, $2/7$ seems like a quite peculiar bound. Next, without going into details of the proofs, we want to give at least some insight why \sout{$2/7$ is} arguably the right and natural error rate for the adaptive setting \newblue{should be strictly higher than $1/3$}.

We first give an intuitive argument why even adaptive protocols cannot deal with
 an error rate of $1/3$. For this, the adversary runs the same strategy as above
 which concentrates all errors on one party only. In particular, given a $3N$ rounds conversation and a budget of $N$ corruptions, the adversary picks one party
, say Alice, and makes her first $N$ transmissions sound like as if her input is $x = 1$. The next $N$ transmissions are made to sound like Alice has input $x = 0$. During the first $N$ responses, regardless of whether $x=1$ (resulting in Alice talking herself) or $x=0$ (resulting in the adversary imitating the same transmissions), the whole conversation will sound legitimate. This prevents any rate adaptation, in this case on Bob's side, to kick in before $2N$ rounds of back and forth have passed. Only then it becomes apparent to the receiver of the corruptions, in this case Bob, that the adversary is trying to fool him. Knowing that the adversary will only try to fool one party, Bob can then stop talking and listen to Alice for the rest of the conversation. Still, even if Bob listens exclusively from this point on, there are only $N$ rounds left which is just enough for all of them to be corrupted. Having received $N$ transmission from Alice claiming $x=1$ and equally many claiming $x=0$, Bob is again left puzzled. This essentially proves the impossibility of tolerating an error rate of $1/3$. But even this $1/3$ error rate \sout{is} \newblue{seems} not achievable. To explain why even a lower fraction of errors\sout{, namely $2/7$,} \newblue{should intuitively} lead to a negative result, we remark that the radical immediate back-off we just assumed for Bob \sout{is} \newblue{seems} not possible. The reason is that if both parties are so sensitive and radical in their adjustment, the adversary can fool both parties simultaneously by simply corrupting a few transmissions of both parties after round $2N$. This would lead to both parties assuming that the transmissions of the other side are being corrupted. The result would be both parties being silent simultaneously, which wastes valuable communication rounds. Choosing the optimal tradeoff for how swift and strong protocols are able to adaptively react without falling into this trap \newblue{should} naturally lead to an error rate strictly between $1/4$ and $1/3$\sout{, and what rate in this range could be more natural than the mediant $2/7$}. \newblue{(Remark: As proven in \cite{efremenko2021optimal} the correct answer is $5/16 \in ]1/4,1/3[$.)}

\paragraph{Other Settings}

We also give results on other settings that have been suggested in the literature, in particular, list decoding and the shared randomness setting of \cite{FGOS13}. We briefly describe these results next. 

Franklin et al. \cite{FGOS13} showed that if both parties share some random string not known to the adversary, then non-adaptive protocols can boost the tolerable error rate from $1/4$ to $1/2$. We show that also in this setting adaptivity helps to increase the tolerable error rate. In particular, in \Cref{sec:sharedrand}, we prove that an error rate of $2/3 - \eps$ is achievable and best possible\footnote{It is interesting to note that similarly to the $2/7$ bound (see \Cref{sec:whyadaptivityhelps}), $2/3$ is the mediant between $1/2$ and $1$, that is, the mediant between the error rate for non-adaptive protocols and the hypothetical error rate of immediately reacting/adapting protocols.}:

\begin{theorem}\label{thm:23boundsoverview}
In the shared randomness setting of \cite{FGOS13}, there exists a efficient robust coding scheme for an error rate of $2/3 - \eps$ while no such scheme exists for an error rate of $2/3$. That is, the equivalents of \Cref{thm:27upperboundoverview} and \Cref{thm:27lowerboundoverview} hold for an error rate of $2/3$. The number of rounds of the robust coding scheme can furthermore be reduced to $O(n)$ if one allows exponential time computations. 
\end{theorem}

We also give the first results for list decodable coding schemes (see \Cref{sec:settingDefinitions} for their definition). The notion of list decodability has been a somewhat recent but already widely successful addition to the study of error correcting codes. It is known that for error correcting codes such a relaxation leads to being able to efficiently \cite{SudanListDecoding} tolerate any constant error rate below $1$, which is a factor of two higher than the $1/2 - \eps$ error rate achievable with unique decoding. It has been an open question whether list decoding can also lead to higher tolerable error rates in interactive coding schemes (see \cite[Open Problem 9]{BravermanAllerton} and \cite[Conclusion]{BR11}). We show that this is indeed the case. In particular, for the non-adaptive setting the full factor of two improvement can also be realized in the interactive setting:

\begin{theorem}\label{thm:12upperboundlistdecodingoverview}
Suppose $\Pi$ is an $n$-round protocol over a constant bit-size alphabet. For any $\eps >0$ there is a $O(1)$-list decodable, non-adaptive, deterministic, and computationally efficient protocol $\Pi'$ that robustly simulates $\Pi$ for an error rate of $1/2 - \eps$ using $O(n^2)$ rounds and an $O(1)$-bit size alphabet.
\end{theorem}
The proof of this theorem is presented in \Cref{sec:AlgLargeRounds}. We also show that the $1/2 - \eps$ error rate is best possible even for adaptive coding schemes. That is, no adaptive or non-adaptive coding scheme can achieve an error rate of $1/2$. We prove these impossibility results formally in \Cref{sec:appendix:furtherlowerbounds}. 

Taken together, our results provide tight negative and matching positive results for any of the eight interactive coding settings given by the three Boolean attributes, \{unique decoding / list decoding\}, \{adaptive / non-adaptive\}, and \{without shared randomness / with shared randomness\} (at least when allowing a linear size alphabet or quadratic number of rounds in the simulation). \Cref{table:bounds} shows the maximum tolerable error rate for each of these settings:\\

\begin{table}[h!]
    \begin{tabular}{l|c|c|c|c|}
    ~            & unique dec. (UD) & UD \& shared rand. & list dec. (LD) & LD \& shared rand. \\ \hline
    Non-adaptive & \ \ \ \ \ \ \ \ \ \ 1/4 \ (\cite{BR11}\ )                                      & \ \ \ \ \ \ \ \ \ \ 1/2 \ (\cite{FGOS13}\ )                                 & 1/2        & 1/2                               \\ \hline
    Adaptive     & \sout{$2/7$} \newblue{[1/3, 2/7]}                                         & 2/3                                      & 1/2                                       & 2/3                               \\ \hline
    \end{tabular}
\caption {Unless marked with a citation all results in this table are new and presented in this paper. Matching positive and negative results for each setting\newblue{, except for the $1/3$ and $2/7$ upper and lower bounds for the adaptive unique decoding setting,} show that the error rates are tight. Even more, the error rates are invariable of assumptions on the communication and computational complexity of the coding scheme and the adversary (see \Cref{hypthesis}, \Cref{cor:weakhyp}, and \Cref{sec:appendix:furtherresultssupportingthehypothesis}).} \label{table:bounds}
\end{table}

\subsection{Invariability Hypothesis: A Path to Natural Interactive Coding Schemes}\label{sec:largerpicture}

In this section, we take a step back and propose a general way to understand the tolerable error rates specific to each setting and to design interactive coding schemes achieving them. We first formulate a strong working hypothesis which postulates that tolerable error rates are invariable regardless of what communication and computational resources are given to the protocol and to the adversary. We then use this hypothesis to determine the tight tolerable error rate for any setting by looking at the simplest setup. Finally, we show how clean insights coming from these simpler setups can lead to designs for intuitive, natural, and easily analyzable interactive coding schemes for the more general setup. 

\paragraph{Invariability Hypothesis}
In this section we formulate our invariability hypothesis.

Surveying the literature for what error rates could be tolerated by different interactive coding schemes, the maximum tolerable error rate appears to vary widely depending on the setting and more importantly, depending on what kind of efficiency one strives for. For example, even for the standard setting---that is, for non-adaptive unique decoding coding schemes using a large alphabet---the following error rates apply: for unbounded (or exponential time) computations, Schulman \cite{Schulman} tolerates a $1/240$ error rate; Braverman and Rao \cite{BR11} improved this to $1/4$; for sub-exponential time computations, Braverman \cite{Braverman12} gave a scheme working for any error rate below $1/40$; for randomized polynomial time algorithms, Brakerski and Kalai \cite{BK12} got an error rate of $1/16$; for randomized linear time computations, Brakerski and Naor \cite{BN13} obtained an unspecified constant error rate smaller than $1/32$; lastly, assuming polynomially bounded protocols and adversaries and using a super-linear number of rounds, Chung et al.~\cite{KnowledgePreservingIC} gave coding schemes tolerating an error rate of $1/6$ (with additional desirable properties). 

We believe that this variety is an artifact of the current state of knowledge rather then being the essential truth. In fact, it appears that any setting comes with exactly one tolerable error rate which furthermore exhibits a strong threshold behavior: For any setting, there seems to be one error rate for which communication is impossible regardless of the resources available, while for error rates only minimally below it simple and efficient coding schemes exist. In short, the tolerable error rate for a setting seems robust and completely independent of any communication resource or computational complexity restrictions made to the protocols or to the adversary. 

Taking this observation as a serious working hypothesis was a driving force in obtaining, understanding, and structuring the results obtained in this work. As we will show, it helped to identify the simplest setup for determining the tolerable error rate of a setting, served as a good pointer to open questions, and helped in the design of new, simple, and natural coding schemes. We believe that these insights and schemes will be helpful in future research to obtain the optimal, and efficient coding schemes postulated to exist. In fact, we already have a number of subsequent works confirming this (e.g., the results of \cite{Braverman12,BK12,KnowledgePreservingIC} mentioned above can all be extended to have the optimal $1/4$ error rate; see also \Cref{sec:appendix:furtherresultssupportingthehypothesis}). All in all, we believe that identifying and formulating this hypothesis is an important conceptual contribution of this work:

\begin{hypothesis}[Invariability Hypothesis]\label{hypthesis}
Given any of the eight settings for interactive communication (discussed above) the maximum tolerable error rates is invariable regardless:
\begin{enumerate}
	\item whether the protocol to be simulated is an arbitrary $n$-round protocol or the much simpler ($n$-bit) exchange protocol, and
	\item whether only $O(1)$-bit size alphabets are allowed or alphabets of arbitrary size, and
	\item whether the simulation has to run in $O(n)$ rounds or is allowed to use an arbitrary number of rounds, and
	\item whether the parties are restricted to polynomial time computations or are computationally unbounded, and
	\item whether the coding schemes have to be deterministic or are allowed to use private randomness (even when only requiring an $o(1)$ failure probability), and
	\item whether the adversary is computationally unbounded or is polynomially bounded in its computations (allowing simulation access to the coding scheme if the coding scheme is not computationally bounded)
\end{enumerate}
\end{hypothesis}

We note that our negative results are already as strong as stipulated by the hypothesis, for all eight settings. The next corollary furthermore summarizes how far these negative results combined with the positive results presented in this work (see \Cref{table:bounds}) already imply and prove two weaker versions of the hypothesis: 

\begin{corollary}\label{cor:weakhyp}
The Invariability Hypothesis holds if one weakens point 3. to ``3'. whether only $O(n)$-bit size alphabets are allowed or alphabets of arbitrary size''. \ 
The Invariability Hypothesis also holds if one weakens point 4. to ``4'. whether the simulation has to run in $O(n^2)$ rounds or is allowed to use an arbitrary number of rounds''.
\end{corollary}

We also refer the reader to \cite{GH13} and \Cref{sec:appendix:furtherresultssupportingthehypothesis} for further (subsequent) results supporting the hypothesis, such as, a proof that the hypothesis holds if point \emph{4.} is replaced by ``\emph{4'. whether the simulation has to run in $O(n (\log^* n)^{O(\log^* n)})$ rounds\footnote{Here $\log^* n$ stands for the iterated logarithm which is defined to be the number of times the natural logarithm function needs to be applied to $n$ to obtain a number less than or equal to $1$. We note that this round blowup is smaller than any constant times iterated logarithm applied to $n$, that is, $2^{O(\log^* n \, \cdot \, \log \log^* n)} = o(\overbrace{\log \log \ldots \log}^{\textit{constant times}} n)$.} or is allowed to use an arbitrary number of rounds}''.

\paragraph{Determining and Understanding Tolerable Error Rates in the Simplest Setting}


Next, we explain how we use the invariability hypothesis in finding the optimal tolerable error rates.




Suppose that one assumes, as a working hypothesis, the invariability of tolerable error rates to hold regardless of the computational setup and even the structure of the protocol to be simulated. Under this premise, the easiest way to approach determining the best error rate is in trying to design robust simulations for the simplest possible two-way protocol, the \emph{Exchange Protocol}. This protocol simply gives each party $n$ bits as an input and has both parties learn the other party's input bits by exchanging them (see also \Cref{sec:exchange}). Studying this setup is considerably simpler. For instance, for non-adaptive protocols, it is easy to see that both parties sending their input in an error correcting code (or for $n=1$ simply repeating their input bit) leads to the optimal robust protocol which tolerates any error rate below $1/4$ but not more. The same coding scheme with applying any ECC list decoder in the end also gives the tight $1/2$ bound for list decoding. For adaptive protocols (both with and without shared randomness), finding the optimal robust $1$-bit exchange protocol was less trivial but clearly still better than trying to design highly efficient coding schemes for general protocols right away. Interestingly, looking at simpler setup actually crystallized out well what can and cannot be done with adaptivity, and why. These insights, on the one hand, lead to the strong lower bounds for the exchange problem but, on the other hand, were also translated in a crucial manner to the same tradeoffs for robustly simulating general $n$-round protocols.

\paragraph{Natural Interactive Coding Schemes}
The invariability working hypothesis was also helpful in finding and formalizing simple and natural designs for obtaining robust coding schemes for general protocols.


Before describing these natural coding schemes we first discuss the element of ``surprise/magic'' in prior works on interactive coding. The existence of an interactive coding scheme that tolerates
a constant error rate is a fairly surprising outcome of the prior works, and remains so even in hindsight. One reason for this is that the simplest way of adding redundancy to a conversation, namely encoding each message via an error correcting code, fails dramatically because the adversary can use its budget non-uniformly and corrupt the first message(s) completely. This derails the interaction completely and makes all further exchanges useless even if no corruptions happens from there on. While prior works, such as \cite{Schulman} or \cite{BR11}, manage to overcome this major challenge, their solution remains a technically intriguing works, both in terms of the ingredients they involve (tree codes, whose existence is again a surprising event) and the recipe for converting the ingredients into a solution to the interactive coding problem. As a result it would appear that the challenge of dealing with errors in interactive communication is an inherently complex task.


In contrast to this, we aim to give an intuitive and natural strategy which lends itself nicely to a simple explanation for the possibility of robust interactive coding schemes and even for why their tolerable error rates are as they are. This strategy simply asserts that if there is no hope to fully trust messages exchanged before one should find ways to put any response into the assumed common context by (efficiently) referring back to what one said before. Putting this idea into a high-level semi-algorithmic description gives the following natural outline for a robust conversation:


\begin{algorithm}
\caption{Natural Strategy for a Robust Conversation (Alice's Side)}
\begin{algorithmic}[1]
\small
\State Assume nothing about the context of the conversation
\Loop
	\State Listen 
	\State $E'_B \gets$ What you heard Bob say last (or so far)
	\State $E_A \gets$  What you said last (or so far)
	\If{$E_A$ and $E'_B$ makes sense together}
		\State Determine the most relevant response $r$
		\State Send the response $r$ \emph{but also} include an (efficient) summary of what you said so far ($E_A$)
	\Else
		\State Repeat what you said last ($E_A$)
	\EndIf
\EndLoop
\State Assume / Output the conversation outcome(s) that seem most likely 
\end{algorithmic}
\label{alg:highlevelSimulator}
\end{algorithm}

At first glance the algorithm may appear vague. In particular notions like ``making sense'', and ``most relevant response'', seem ambiguous and subject to interpretation. It turns out that this is not the case. In fact, formalizing this outline into a concrete coding scheme turns out to be straight forward. This is true especially if one accepts the invariability working hypothesis and allows oneself to not be overly concerned with having to immediately get a highly efficient implementation. In particular, this permits to use the simplest (inefficiently) summary, namely referring back word by word to everything said before. This straight-forward formalization leads to Algorithm \ref{alg:Simulator}. Indeed, a reader that compares the two algorithms side-by-side will find that Algorithm \ref{alg:Simulator} is essentially a line-by-line formalization of Algorithm \ref{alg:highlevelSimulator}.

In addition to being arguably natural, Algorithm \ref{alg:Simulator} is also easy to analyze. Simple counting arguments show that the conversation outcome output by most parties is correct if the adversary interferes at most a $1/4 - \eps$ fraction of the time, proving the tight tolerable error rate for the robust (while somewhat still inefficient) simulation of general $n$-round protocols. Maybe even more surprisingly, the exact same simple counting argument also directly shows our list decoding result, namely, that even with an error rate of $1/2 - \eps$ the correct conversation outcome is among the $1/\eps$ most likely choices for both parties. Lastly, it is easy to enhance both Algorithm \ref{alg:highlevelSimulator} and similarly Algorithm \ref{alg:Simulator} to be adaptive. For this one simply adds the following three, almost obvious, rules of an adaptive conversation:

\begin{rules}[Rules for a Successful Adaptive Conversation]\label{rules:ForAnAdaptiveConversation} \mbox{}\\
Be fair and take turns talking and listening, unless:
\begin{enumerate}
	\item you are sure that your conversation partner already understood you fully and correctly, in which case you should stop talking and instead listen more to also understand him; or reversely
	\item you are sure that you already understood your conversation partner fully and correctly, in which case you should stop listening and instead talk more to also inform him.
\end{enumerate}
\end{rules}

Our algorithm Algorithm \ref{alg:SimulatorAdaptive} exactly adds the formal equivalent of these rules to Algorithm \ref{alg:SimulatorAdaptive}. A relatively simple proof that draws on the insights obtained from analyzing the optimal robust exchange protocol then shows that this simple and natural coding scheme indeed achieves the \sout{optimal} $2/7 - \eps$ tolerable error rate for adaptive unique-decoding. This means that Algorithm \ref{alg:SimulatorAdaptive} is one intuitive and natural algorithm that simultaneously achieves the $1/4$ error rate (if the adaptivity rules are ignored), the $2/7 - \eps$ error rate for adaptive protocols and the $1/2 - \eps$ error rate with optimal list size when list decoding is allowed. Of course, so far, this result comes with the drawback of using a large ($O(n)$-bits) alphabet. Nonetheless, this result together with the invariability hypothesis holds open the promise of such a ``perfect'' algorithm that works even without the drastic communication overhead\footnote{Subsequent works of the authors have already moved towards this constant size alphabet protocol by reducing the alphabet size from $O(n)$ bits to merely $O(\log^{\eps} n)$ bits (see \Cref{sec:appendix:furtherresultssupportingthehypothesis}).}.

\section{Results for the Exchange Problem}\label{sec:exchange}
In this section we study the \emph{Exchange Problem}, which can be viewed as the simplest instance of a two-way (i.e., interactive) communication problem. In the Exchange Problem, each party is given a bit-string of $n$ bits, that is, $i_A, i_B \in \{0,1\}^{n}$, and each party wants to know the bit-string of the other party.

Recall that the $1/4$ impossibility bound on tolerable error-rate for non-adaptive interactive protocols presented by Braverman and Rao~\cite{BR11} is this simple setting. In \Cref{sec:27rateadaptationXOR}, we show that adding rate adaptivity to the exchange algorithms helps one break this $1/4$ impossiblity bound and tolerate an error-rate of $2/7-\eps$, and in fact, this is done with a minimal amount of adaptivity-based decisions regarding whether a party should transmit or listen in each round. 
We show in \mbox{\Cref{subsec:xor13impossible}}  that the error-rate of \sout{$2/7$} \newblue{$1/3$} is not tolerable even for the exchange problem, even if one is given infinite number of rounds, alphabet size, and computational power. Furthermore, the intuition used to achieve the $2/7-\eps$ possibility result also extends to the more general \emph{simulation} problem, discussed in \mbox{\Cref{sec:simulators}}.

\subsection{An Algorithm for the Exchange Problem under Error-Rate $2/7-\eps$}\label{sec:27rateadaptationXOR}
Note that a simple solution based on error correcting codes suffices for solving exchange problem under error-rate $\frac{1}{4}-\eps$: parties use a code with relative distance $1-\eps$. In the first half of the time, Alice sends its encoded message and in the second half of the time, Bob sends its encoded message. At the end, each party decodes simply by picking the codeword closest to the received string. As discussed before, the error rate $\frac{1}{4}-\eps$ of this approach is the best possible if no rate adaptation is used. In the following, we explain that a simple rate adaptation technique boosts the tolerable error-rate to $\frac{2}{7}-\eps$, which we later prove to be optimal.

\begin{theorem}\label{thm:27xor}
In the private randomness model with rate adaptation, there is an algorithm for the $n$-bit Exchange Problem that tolerates adversarial error rate of $2/7 - \eps \approx 0.2857-\eps$ for any $\eps>0$.
\end{theorem}



\begin{proof}
The algorithm runs in $N=7n/\eps$ rounds, which means that the budget of adversary is $(2/7-\eps) 7n/\eps=2n/\eps-7$.  Throughout the algorithm, we use an error-correction code $\mathcal{C}: \{0,1\}^{n} \rightarrow \{1, \dots, q\}^{\frac{n}{\eps}} $ that has distance $\frac{n}{\eps}-1$. Also, for simplicity, we use $\mathcal{C}^\kappa$ to denote the code formed by concatenating $\kappa$ copies of $\mathcal{C}$. 

The first $6N/7$ rounds of the algorithm do not use any rate adaptation: Simply, Alice sends $\mathcal{C}^3(i_A)$ in the first $3N/7$ rounds and Bob sends $\mathcal{C}^3(i_B)$ in the second $3N/7$ rounds. At the end of this part, each party ``estimates'' the errors invested on the transmissions of the other party by simply reading the hamming distance of the received string to the closest codeword of code $\mathcal{C}^3$. If this estimate is less than $N/7=n/\eps$, the party---say Alice---can safely assume that the closest codeword is the correct codeword. This is because the adversary's total budget is $2n/\eps-7$ and the distance between two codewords of $\mathcal{C}^3(i_B)$ is at least $3n/\eps-3$. In this case, in the remaining $N/7$ rounds of the algorithm, Alice will be sending $\mathcal{C}(i_A)$ and never listening. On the other hand, if Alice reads an estimated error greater than $N/7=n/\eps$, then in the remaining $N/7$ rounds, she will be always listening. The algorithm for Bob is similar. 

Note that because of the limit on the budget of the adversary, at most only one of the parties will be listening in the last $N/7$ rounds. Suppose that there is exactly one listening party and it is Alice. In this case, throughout the whole algorithm, Alice has listened a total of $4N/7=4n/\eps$ rounds where Bob has been sending $\mathcal{C}^4(i_B)$. Since the adversaries budget is less than $2n/\eps-7$, and because $\mathcal{C}^4$ has distance $\frac{4n}{\eps}-4$, Alice can also decode correctly by just picking the codeword of $\mathcal{C}^4$ with the smallest hamming distance to the received string. 

\end{proof}


\subsection{Impossibility of Tolerating Error-Rate \sout{$2/7$} {\color{blue} $1/3$} in the Exchange Problem}\label{subsec:xor13impossible}

In this section, we turn our attention to impossibility results and particularly prove that the error-rate \sout{$2/7$} {\color{blue} $1/3$} is not tolerable. See the formal statement in \mbox{\Cref{thm:27lowerboundoverview}}.

Note that Braverman and Rao~\cite{BR11} showed that it is not possible to tolerate error-rate of $1/4$ with non-adaptive algorithms. For completeness, we present a (slightly more formal) proof in the style of distributed indistinguishably arguments, in \Cref{sec:Impossibilities}, which also covers random algorithms.

We first explain a simple (but informal) argument which shows that even with adaptivity, error-rate $1/3$ is not tolerable. A formal version of this proof is deferred to \Cref{sec:Impossibilities}. \sout{The informal version explained next serves as a warm up for the more complex argument used for proving the $2/7$ impossibility, presented formally in \mbox{\Cref{thm:xor27impossible}}.}

\begin{figure}[t]
	\centering
		\includegraphics[width=0.5\textwidth]{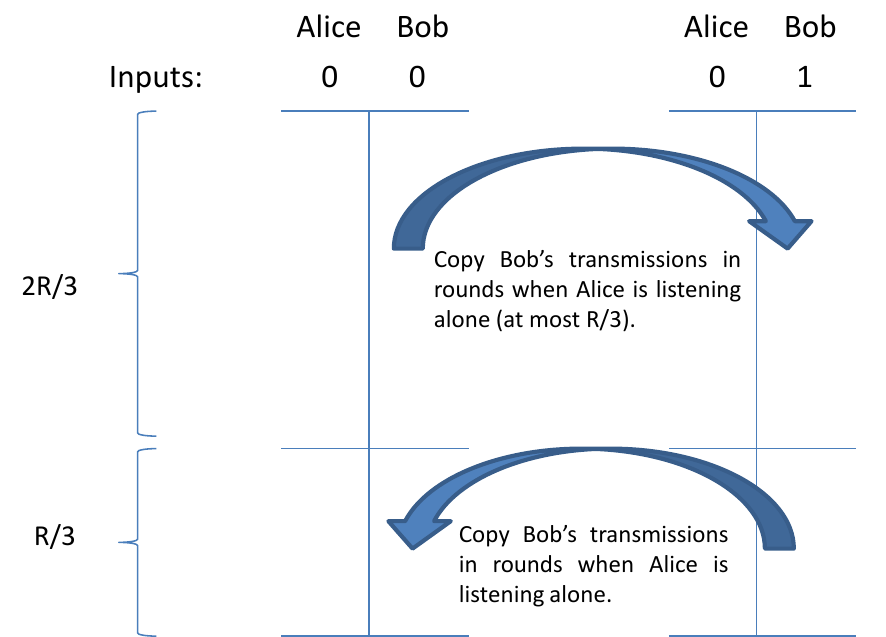}
	\caption{ The adversary's strategy for the $1/3$-impossibility proof}
	\label{fig:OneThird-LB}
\end{figure}

\begin{lemma}\label{lem:xor13impossible}
There is no (deterministic or randomized) adaptive protocol that robustly simulates the Exchange Protocol for an error rate of $1/3$ with an $o(1)$ failure probability even when allowing computationally unbounded protocols that use an arbitrarily large number of rounds and an unbounded alphabet.
\end{lemma}
\begin{proof}[Informal Proof]
To simplify the discussion, here we only explain the reasoning about deterministic algorithms and we also ignore the rounds in which both parties listen. Note that by the definition of the model, in those all-listening rounds, the adversary can deliver arbitrary messages to each of the parties at no cost. A complete proof is presented in \Cref{sec:Impossibilities}. 

For the sake of contradiction, suppose that there is an algorithm that solves the exchange problem under adversarial error-rate $1/3$, in $N$ rounds. We work simply with $1$-bit inputs. Let $S_{X,Y}$ denote the setting where Alice receives input $X$ and Bob gets input $Y$. The idea is to make either settings $S_{0,0}$ and $S_{0,1}$ look indistinguishable to Alice or settings $S_{0,0}$ and $S_{1,0}$ look indistinguishable to Bob. 

Consider setting $S_{0,0}$ and suppose that for the first $2N/3$ rounds, the adversary does not interfere. Without loss of generality, we can assume that in this setting, Alice listens (alone) in less than $N/3$ of these $2N/3$ rounds. We next explain adversary's strategy for the case that this assumption holds. An illustration of the adversary's strategy is presented in \Cref{fig:OneThird-LB}.

First, we explain the adversary's strategy for setting $S_{0,1}$: Adversary creates a \emph{dummy personality} $\widetilde{Bob}_0$ and simulates it with Alice in setting $S_{0,0}$ where adversary does not interfere. In the first $2N/3$ rounds of setting $S_{0,1}$, whenever Alice listens (alone), the adversary delivers transmission of $\widetilde{Bob}_0$ to Alice. As a shorthand for this, we say \emph{Alice is connected to $\widetilde{Bob}_0$}. Since Alice listens less than $N/3$ of the time, the adversary will have enough budget to completely fake Bob as $\widetilde{Bob}_0$ (from Alice's viewpoint). Thus, the two settings look identical to Alice for the first $2N/3$ rounds. During the last $N/3$ rounds of the execution in setting $S_{0,1}$, the adversary lets Alice and the real Bob talk without no interference. 

Now, we explain the adversary's strategy for setting $S_{0,0}$: The adversary generates another dummy personality $\widetilde{Bob}_1$ by simulating Bob in setting $S_{0,1}$ where alone-listening rounds of Alice in the first $2N/3$ rounds are connected to $\widetilde{Bob}_0$. In setting $S_{0,0}$, the adversary lets Alice and Bob talk freely during the first $2N/3$ rounds but for the last $N/3$ rounds, whenever Alice listens, the adversary connects her to $\widetilde{Bob}_1$. 

To conclude, we know that in each of the settings $S_{0,1}$ and $S_{1,0}$, at most $N/3$ rounds get corrupted by the adversary. Furthermore, the two settings look identical to Alice which means that she can not know Bob's input. This contradiction completes the proof.
\end{proof}

{\color{red} {\color{blue} The rest of this  section, indicated in {\color{red} red}, is withdrawn.}

\medskip
\begin{theorem}\label{thm:xor27impossible}[A rephrasing of \Cref{thm:27lowerboundoverview}]
There is no (deterministic or randomized) adaptive protocol that robustly simulates the Exchange Protocol for an error rate of $2/7$ with an $o(1)$ failure probability even when allowing computationally unbounded protocols that use an arbitrarily large number of rounds and an unbounded alphabet.
\end{theorem}

\begin{proof}
Suppose that there is an algorithm that solves the exchange problem under adversarial error-rate $1/3$, in $N$ rounds. We study this algorithm simply with $1$-bit inputs. Let $S_{X,Y}$ denote the setting where Alice receives input $X$ and Bob gets input $Y$. We specifically work only with settings $S_{0,0}$, $S_{0,1}$, and $S_{1,0}$. Note that if a party has an input $1$, it knows in which of these three settings we are. The idea is to present an adversarial strategy that changes the receptions of the party, or parties, that have a $0$ input so as to make that party, or parties, to not be able to distinguish (between two of) the settings. 

For simplicity, we first assume that the algorithm is deterministic and we also ignore the rounds where both parties listen. Note that by the definition of the model for adaptive algorithms (see ~\Cref{sec:settingDefinitions}), for these rounds, the adversary can deliver arbitrary messages to the parties at no cost. 

For this lower bound, we need to define the party that becomes the base of indistinguishability (whom we confuse by errors) in a more dynamic way, compared to that in \Cref{lem:xor13impossible} or in \cite[Claim 10]{BR11}. For this purpose, we first study the parties that have input $0$ under a particular pattern of received messages (regardless of the setting in which they are), without considering whether the adversary has enough budget to create this pattern or not. Later, we argue that the adversary indeed has enough budget to create this pattern for at least one party and make that party confused. 

To determine what should be delivered to each party with input $0$, the adversary cultivates \emph{dummy personalities} $\widetilde{Alice}_0$, $\widetilde{Alice}_1$, $\widetilde{Bob}_0$, $\widetilde{Bob}_1$, by simulating Alice or Bob respectively in settings $S_{0,0}$, $S_{1,0}$, $S_{0,0}$, and $S_{0,1}$, where each of these settings is modified by adversarial interferences (to be specified). Later, when we say that in a given round, e.g. ``the adversary \emph{connects} dummy personality $\widetilde{Bob}_1$ to Alice", we mean that the adversary delivers the transmission of $\widetilde{Bob}_1$ in that round to Alice\footnote{This is assuming $\widetilde{Bob}_1$ transmits in that round, we later discuss the case where both Alice and $\widetilde{Bob}_1$ listen later.}. For each setting, the adversary uses one method of interferences, and thus, when we refer to a setting, we always mean the setting with the related adversarial interferences included.  
 
We now explain the said pattern of received messages. Suppose that Alice has input $0$ and consider her in settings $S_{0, 0}$ and $S_{0,1}$, as a priori these two settings are identical to Alice. Using connections to \emph{dummy personalities}, the adversary creates the following pattern: In the first $2N/7$ rounds in which Alice listens alone, her receptions will be made to imply that Bob also has a $0$. This happens with no adversarial interference in setting $S_{0,0}$, but it is enforced to happen in setting $S_{0,1}$ by the adversary via connecting to Alice the dummy personality $\widetilde{Bob}_0$ cultivated in setting $S_{0,0}$. Thus, at the end of those $2N/7$ listening-alone rounds of Alice, the two settings are indistinguishable to Alice. In the rest of the rounds where Alice listens alone, the receptions will be made to look as if Bob has a $1$. That is, the adversary leaves those rounds of setting $S_{0,1}$ intact, but in rounds of setting $S_{0,0}$ in which Alice listens alone, the adversary connects to Alice the dummy personality $\widetilde{Bob}_1$ cultivated in setting $S_{0,1}$ (with the adversarial behavior described above). 

The adversary creates a similar pattern of receptions for Bob when he has an input $0$, in settings $S_{0,0}$ and $S_{1,0}$. That is, the first $2N/7$ of his alone receptions are made to imply that Alice has a $0$ but the later alone-receptions imply that Alice has a $1$.

The described reception pattens make Alice unable to distinguish $S_{0,0}$ from $S_{0,1}$ and also they make Bob unable to distinguish $S_{0,0}$ from $S_{1,0}$. However, the above discussions ignore the adversary's budget. We now argue that the adversary indeed has enough budget to create this reception pattern to confuse one or both of the parties. 

Let $x_A$ be the total number of rounds where Alice listens, when she has input $0$ and her receptions follow the above pattern. Similarly, define $x_B$ for Bob. If $x_A \leq 4N/7$, then the adversary indeed has enough budget to make the receptions of Alice in settings $S_{0,0}$ and $S_{0,1}$ follow the discussed behavior, where the first $2N/7$ alone-receptions of Alice are altered in $S_{0,1}$ and the remaining alone-receptions are altered in $S_{0,0}$. Thus, if $x_A \leq 4N/7$, the adversary has a legitimate strategy to make Alice confused between $S_{0,0}$ and $S_{0,1}$. A similar statement is true about Bob: if $x_B \leq 4N/7$, the adversary has a legitimate strategy to make Bob confused between $S_{0,0}$ and $S_{0,1}$. 

Now suppose that $x_A> 4N/7$ and $x_B >4N/7$. In this case, the number of alone-receptions of Alice is at most $x_A - (x_A+x_B-N) = N - x_B \leq 3N/7$ and similarly, the number of alone-receptions of Bob is at most $x_B - (x_A+x_B-N) = N - x_A \leq 3N/7$. This is because $x_A+x_B-N$ is a lower bound on the overlap of the round that the two parties listen. In this case, the adversary has enough budget to simultaneously confuse both Alice and Bob of setting $S_{0,0}$; Alice will be confused between $S_{0,0}$ and $S_{0,1}$ and Bob will be confused between $S_{0,0}$ and $S_{1,0}$. For this purpose, in setting $S_{0,0}$, the adversary leaves the first $2N/7$ alone-receptions of each party intact but alters the remaining at most $N/7$ alone-receptions of each party, for a total of at most $2N/7$ errors. On the other hand, in setting $S_{0,1}$, only $2N/7$ errors are used on the first $2N/7$ alone-receptions of Alice and similarly, in setting $S_{1,0}$, only $2N/7$ errors are used on the first $2N/7$ alone-receptions of Bob. Note that these errors make the receptions of each party that has input $0$ follow the pattern explained above. 

We now go back to the issue of the rounds where both parties listen. The rounds of $S_{0,0}$ in which both parties listen are treated as follows: The adversary delivers the transmission of $\widetilde{Bob}_1$ (cultivated in setting $S_{0, 1}$) to Alice and delivers the transmission of $\widetilde{Alice}_1$ (cultivated in setting $S_{1, 0}$) to Bob. Recall that the adversary does not pay for these interferences. Furthermore, note that these connections make sure that these all-listening rounds do not help Alice to distinguish $S_{0,0}$ from $S_{0,1}$ and also they do not help Bob to distinguish $S_{0,0}$ from $S_{1,0}$.

Finally, we turn to covering the randomized algorithms. Note that for this case we only show that the failure probability of the algorithm is not $o(1)$ as just by guessing randomly, the two parties can have success probability of $1/4$. 

First suppose that $Pr[x_A \leq 4N/7] \geq 1/3$. Note that the adversary can easily compute this probability, or even simpler just get a $(1+o(1))$-factor estimation of it. If $Pr[x_A \leq 4N/7] \geq 1/3$, then the adversary will hedge his bets on that $x_A \leq 4N/7$, and thus, it will try to confuse Alice. In particular, he gives Alice an input $0$ and tosses a coin and gives Bob a $0$ or a $1$, accordingly. Regarding whether $Bob$ gets input $0$ or $1$, the adversary also uses the dummy personalities $\widetilde{Bob}_0$ and $\widetilde{Bob}_1$, respectively. With probability $1/3$, we will have that in fact $x_A\leq 4N/7$, and in this case the adversary by determining whether Alice hears from the real Bob or the dummy Bob, the adversary makes Alice receive the messages with the pattern described above. This means Alice would not know whether Bob has a $0$ or a $1$. Hence, the algorithm fails with probability at least $1/6$ (Alice can still guess in this case which is correct with probability $1/2$).
Similarly, if $Pr[x_B \leq 4N/7] \geq 1/3$, then adversary will  make Bob confused between $S_{0,0}$ and $S_{1,0}$. On the other hand, if $Pr[x_A \leq 4N/7] < 1/3$ and $Pr[x_B \leq 4N/7] < 1/3$, then just using a union bound we know that $\Pr[x_A> 4N/7 \& x_B >4N/7] \geq 1/3$. In this case, the adversary gambles on the assumption that it will actually happen that $x_A> 4N/7$ and $x_B >4N/7$. This assumption happens with probability at least $1/3$, and in that case, the adversary makes Alice confused between $S_{0,0}$ and $S_{0,1}$ and Bob confused between $S_{0,0}$ and $S_{1,0}$, simultaneously, using the approach described above. Hence, in conclusion, in any of the cases regarding random variables $x_A$ and $x_B$, the adversary can make the algorithm fail with probability at least $1/6$. 

\end{proof}
}

\section{Natural Interactive Coding Schemes With Large Alphabets}\label{sec:simulators}
We start by presenting a canonical format for interactive communication and then present our natural non-adaptive and adaptive coding schemes.

\subsection{Interactive Protocols in Canonical Form}
We consider the following canonical form of an $n$-round two party protocol over alphabet $\Sigma$: We call the two parties Alice and Bob. To define the protocol between them, we take a rooted complete $|\Sigma|$-ary tree of depths $n$. Each non-leaf node has $|\Sigma|$ edges to its children, each labeled with a distinct symbol from $\Sigma$. For each node, one of the edges towards children is \emph{preferred}, and these \emph{preferred} edges determine a unique leaf or equivalently a unique path from the root to a leaf. We say that the set $\mathcal{X}$ of the preferred edges at odd levels of the tree is owned by Alice and the set $\mathcal{Y}$ of the preferred edges at even levels of the tree is owned by Bob. This means that at the beginning of the protocol, Alice gets to know the preferred edges on the odd levels and Bob gets to know the preferred edges on the even levels. The knowledge about these preferred edges is considered as inputs $\mathcal{X}$ and $\mathcal{Y}$ given respectively to Alice and Bob. The output of the protocol is the unique path from the root to a leaf following only preferred edges. We call this path the \emph{common path} and the edges and nodes on this path the \emph{common edges} and the \emph{common nodes}. The goal of the protocol is for Alice and Bob to determine the common path. The protocol succeeds if and only if both Alice and Bob learn the common path. Figure \ref{fig:PointerJumping} illustrates an example: Alice's preferred edges are indicated with blue arrows and Bob's preferred edges are indicated with red arrows, and the common leaf is indicated by a green circle. 

\begin{figure}[t]
	\centering
		\includegraphics[width=0.6\textwidth]{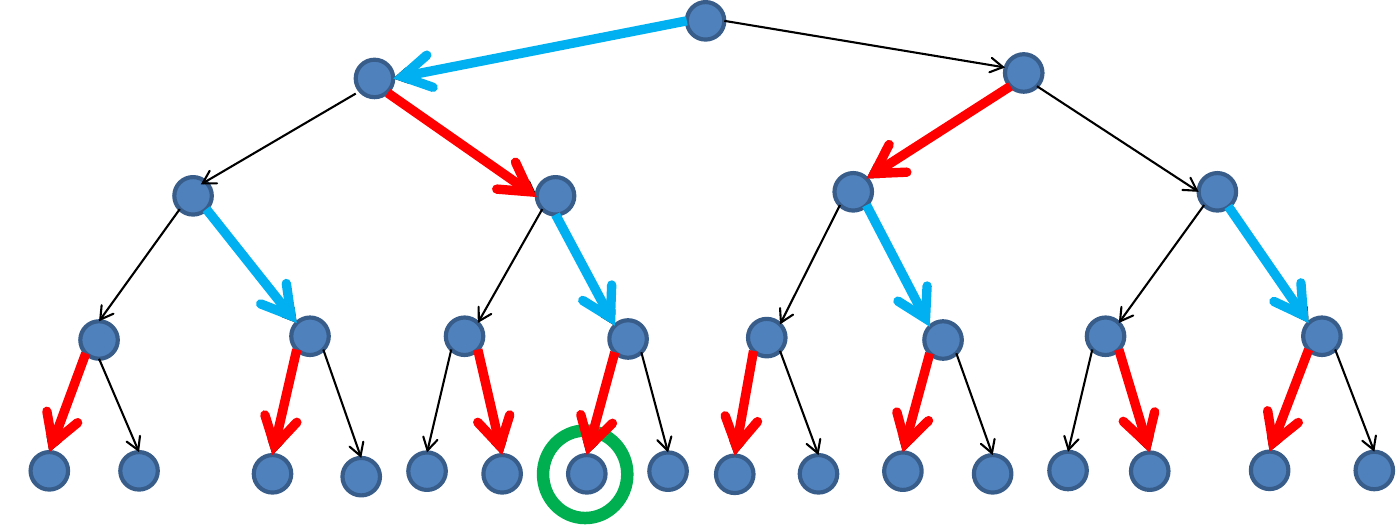}
	\caption{A Binary Interactive Protocol in the Canonical Form}
	\label{fig:PointerJumping}
\end{figure}

It is easy to see that if the channel is noiseless, Alice and Bob can determine the common path of a canonical protocol $P$ by performing $n$ rounds of communication. For this Alice and Bob move down on the tree together simply by following the path of preferred edges; they take turns and exchange one symbol of $\Sigma$ per round, where each symbol indicates the next common node. We call this exchange the execution of the protocol $P$. 


\subsection{Natural Non-Adaptive Coding Schemes}\label{sec:alg}
In this section, we present a non-adaptive coding scheme which can be viewed as a straightforward formalization of the natural high level approach presented in \Cref{sec:overview}. This coding scheme tolerates the optimal error rate of $1/4-\eps$ when unique decoding and simultaneously the optimal error rate of $1/2-\eps$ when list decoding. The coding scheme is furthermore simple, intuitive, and computationally efficient, but it makes use of a large $O(\frac{n}{\eps})$-bit size alphabet. We note that one can also view this algorithm as a simplified version of the Braverman-Rao algorithm~\cite{BR11} with larger alphabet size and without using tree codes~\cite{Schulman}. 

The algorithm, for which a pseudo code is presented in Algorithm \ref{alg:Simulator}, works as follows: In the course of the algorithm, Alice and Bob respectively maintain sets $E_A$ and $E_B$ which are a subset of their own preferred tree edges that are considered to be \emph{important}. We call these \emph{important edge-sets} or sometimes simple \emph{edge-sets}. Initially these edge-sets are empty and in each iteration, Alice and Bob add one edge to their sets. In each iteration, when a party gets a turn to transmit, it sends its edge-set to the other party. The other party receives either the correct edge-set or a corrupted symbol which represents an edge-set made up by the adversary. In either case, the party combines the received edge-set with its own important edge-set and follows the common path in this set. Then, if this common path can be extended by the party's own set of \emph{preferred edges} by a new edge $e$, the party adds this edge $e$ to its edge-set, and sends this new edge-set in the next round. If, on the other hand, the common path already ends at a leaf, then the party registers this as a vote for this leaf and simply re-sends its old edge-set. In the end, both parties simply output the the leaf (respectively the $O(1/\eps)$ leaves) with the most votes for unique decoding (resp., for list decoding).

\begin{algorithm}[t]
\caption{Natural Non-Adaptive Coding Scheme at Alice's Side}
\begin{algorithmic}[1]
\small
\State $\mathcal{X} \gets$ the set of Alice's preferred edges; 
\State $E_A \gets \emptyset$; \Comment{$E_A$ is Alice's set of \emph{important} edges. We preserve that always $E_A \subseteq \mathcal{X}$}
\State $N \gets \frac{2n}{\eps}$; 
\For{$i=1$ to N}
	\State Receive edge-set $E'_B$; \Comment{$E'_B$ is the received version of Bob's important edge-set $E_B$}
	\State $E \gets E'_B \cup E_A$
	\If{$E$ is a valid edgeset} 
		\State $r \gets \emptyset$
		\State follow the common path in $E$ 
		\If{the common path ends at a leaf}
			\State Add one vote to this leaf
		\Else
			\State $r \gets \{e\}$ where $e$ is the next edge in $\mathcal{X}$ continuing the common path in $E$ (if any)
    \EndIf
		\State $E_A \gets E_A \cup r$
		\State Send $E_A$
	\Else
		\State Send $E_A$
	\EndIf	
\EndFor
\State Output the leaf with the most votes for unique decoding
\State Output the $O(1/\eps)$ leaves with the most votes for list decoding
\end{algorithmic}
\label{alg:Simulator}
\end{algorithm}

\paragraph{Analysis}
We now prove that Algorithm \ref{alg:Simulator} indeed achieves the optimal tolerable error rates for non-adaptive unique decoding and list decoding.

\begin{theorem}\label{thm:nonadaptiveSim}
For any $\eps>0$, Algorithm \ref{alg:Simulator} is a deterministic polynomial time non-adaptive simulator with alphabet size of $O(\frac{n}{\eps})$-bits and round complexity  $\frac{2n}{\eps}$ that tolerates an error-rate of $1/4-\eps$ for unique decoding, and also tolerates an error-rate of $1/2 - \eps$ for list decoding with a list of size $\frac{1}{\eps}$.
\end{theorem}
\begin{proof}
Clearly, both $E_A$ and $E_B$ grow by at most one edge per round. Furthermore, the edges always attach to an already present edge and therefore, each of these edge-sets always forms a subtree with size at most $N$starting at the root of the tree of the canonical form, which has depth $n$. One can easily see that each such subtree can be encoded using $O(N)$ bits, e.g., by encoding each edge of the breadth first search traversal of the subtree using alphabet of size $3$ (indicating ``left", ``right" or ``up"). Hence, parties can encode their edge-sets using $O(\frac{n}{\eps})$-bits symbols, which shows that the alphabet size is indeed as specified.

We now prove the correctness of the algorithm, starting with that of unique decoding. Note that any two consecutive rounds in which Bob and Alice's transmissions are not corrupted by adversary, one of the following two good things happens: Either the path in $E_A \cup E_B$ gets extended by at least one edge, or both Alice and Bob receive a vote for the correct leaf. 

Now suppose that the simulation runs in $N = 2n/\eps$ rounds which can be grouped into $n/\eps$ round pairs. Given the error rate of $1/4 - \eps$ at most a $1/2 - 2\eps$ fraction of these round pairs can be corrupted, which leaves a total of $N/2 (1/2 + 2\eps)$ uncorrupted round pairs. At most $n$ of these round pairs grow the path while the remaining $N/2 (1/2 + 2\eps) - n$ rounds vote for the correct leaf. This results in at least $N(1/2 + 2\eps) - \ceil{n/2} = \frac{n}{2\eps} + 2n - n > N/4$ out of $N/2$ votes being correct. 

For the list decoding, with error rate $1/2-\eps$, we get that at most $1-2\eps$ fraction of round-pairs are corrupted, and thus at least $N\eps= 2n$ uncorrupted pairs exist. Hence, the correct leaf gets a vote of at least $2n-n$. Noting that the total number of the votes that one party gives to its leaves is $N/2=\frac{n}{\eps}$, we get that the correct leaf has at least a $\eps$ fraction of all the votes. Therefore, if we output the $1/\eps$ leaves with the most votes, the list will include the correct leaf.
\end{proof}

\subsection{Natural Adaptive Coding Scheme}
In this section we show that the simplest way to introduce adaptation into the natural coding scheme presented in Algorithm \ref{alg:Simulator}. In particular we use the simple rules specified as \Cref{rules:ForAnAdaptiveConversation} and show that this leads to a coding scheme tolerating an error rate of $2/7 - \eps$, the optimal error rate for this setting. 

\begin{algorithm}[t]
\caption{Natural Adaptive Coding Scheme at Alice's Side}
\begin{algorithmic}[1]
\small
\State $\mathcal{X} \gets$ the set of Alice's preferred edges; 
\State $E_A \gets \emptyset$; 
\State $N \gets \Theta(\frac{n}{\eps})$; 
\For{$i=1$ to $\frac{6}{7}N$}
	\State Receive edge-set $E'_B$; 
	\State $E \gets E'_B \cup E_A$
	\If{$E$ is a valid edgeset} 
		\State $r \gets \emptyset$
		\State follow the common path in $E$ 
		\If{the common path ends at a leaf}
			\State Add one vote to this leaf
		\Else
			\State $r \gets \{e\}$ where $e$ is the next edge in $\mathcal{X}$ continuing the common path in $E$ (if any)
    \EndIf
		\State $E_A \gets E_A \cup r$
		\State Send $E_A$
	\Else
		\State Send $E_A$
	\EndIf	
\EndFor
\State Let $s$ be number of votes of the leaf with the most votes and $t$ be the total number of votes 
\If{$s\geq t-\frac{N}{7}$}
	\For{$i=1$ to $\frac{N}{7}$}
		\State Send $E_A$
	\EndFor
\Else
	\For{$i=1$ to $\frac{N}{7}$}
		\State Receive edge-set $E'_B$; $E = E'_B \cup E_A$
		\If{$E$ is a valid edge-set} 
		\State follow the common path in $E$ 
		\If{the common path ends at a leaf}
			\State Add one vote to this leaf
		\EndIf
	\EndIf
	\EndFor
\EndIf
\State Output the leaf with the most votes 
\end{algorithmic}
\label{alg:SimulatorAdaptive}
\end{algorithm}


Next we explain how to incorporate the rules specified in \Cref{rules:ForAnAdaptiveConversation} easily and efficiently into Algorithm \ref{alg:Simulator}. For this we note that for example if one party has a leaf with more than $(2/7 - \eps)N$ votes, since adversary has only budget of $(2/7 - \eps)N$, this leaf is the correct leaf and thus the party can follow the second rule. Generalizing this idea, we use the rule that, if the party has a leaf $v$ such that only at most $\frac{N}{7}$ votes are on leaves other than $v$, then the party can safely assume that this is the correct leaf. In our proof we show that this assumption is indeed safe and furthermore, at least one party can safely decode at the end of the first $6/7$ fraction of the simulation. Since both parties know this in advance, if a party can not safely decode after $6/7$ fraction of the time, it knows that the other party has safely decoded---which corresponds to the condition in the first rule--- and in this case, the party only listens for the last $1/7$ fraction of the protocol. The pseudo code for this coding scheme is presented in Algorithm \ref{alg:SimulatorAdaptive}.

\begin{theorem}\label{thm:uniquedecoding}
Algorithm \ref{alg:SimulatorAdaptive} is a deterministic adaptive coding scheme with alphabet size of $O(\frac{n}{\eps})$-bits, round complexity of $O(\frac{n}{\eps})$, and polynomial computational complexity that tolerates an error-rate of $2/7-\eps$ for unique decoding.
\end{theorem}
\begin{proof}
First, we show that if at the end of $\frac{6N}{7}$ rounds, one party has $t$ votes, $s\geq t-\frac{N}{7}$ of which are dedicated to one leaf $v$, then this party can safely assume that this leaf $v$ is the correct leaf. Proof of this part is by a simply contradiction. Note that if the party has $s$ votes, then there are at least $\frac{3N}{7}-t$ that either stopped the growth of the path or turned an edge-set into a nonvalid edge-set. Furthermore, if $v$ is not the correct leaf, then the votes $v$ are created by errors of adversary which means that adversary has invested $s$ errors on turning the edge-sets sent by the other party into other valid-looking edge-sets. Hence, in total, adversary has spent at least $\frac{3N}{7}-t+s\geq \frac{3N}{7}-t+t-\frac{N}{7} \geq \frac{2N}{7}$ errors which is a contradiction.

Now that we know that the rule for safely decoding at the end of $\frac{6N}{7}$ rounds is indeed safe, we show that at least one party will safely decode at that point of time. Suppose that no party can decode safely. Also assume that Alice has $t_{A}$ votes, $r_{A}$ of which are votes on the good leaf. That means at least adversary has turned at least $t_{A}-r_{A}$ edge-sets sent by Bob into other valid-looking edge-sets. Similarly, $t_{B}-r_{B}$ errors are introduced by the adversary on edge-sets sent by Alice. If neither Alice nor Bob can decode safely, we know that $t_{A}-r_{A}\geq \frac{N}{7}$ and $t_{B}-r_{B}\geq \frac{N}{7}$, which means that in total, adversary has introduced at least $\frac{2N}{7}$ errors. Since this is not possible give adversary's budget, we conclude that at the end of $\frac{6N}{7}$ rounds, at least one party decodes safely.

Now suppose that only one party, say Alice, decodes safely at the end of $\frac{6N}{7}$ rounds. Then, in the last $\frac{N}{7}$ rounds, Bob is listening and Alice is sending. In this case, we claim that Bob's leaf that gets the majority of the votes at the end is the correct leaf. The reason is, suppose that Bob has $t_{B}$ votes from the first $\frac{6N}{7}$ rounds and $t'_{B}$ votes from the last $\frac{N}{7}$ rounds. Furthermore, suppose that the correct leaf had $r_{B}$ votes from the first $\frac{6N}{7}$ rounds and $r'_{B}$ votes from the last $\frac{N}{7}$ rounds. Then, the adversary has introduced at least $(\frac{3N}{7}-t_{B})+(\frac{N}{7}-t'_{B})+ (t_{B}-r_{B})+(t'_{B}-r'_{B})= \frac{4N}{7}-r_{B}+r'_{B}$ errors. Since adversaries budget is at most $(\frac{2}{7}-\eps)N$, we get that $r_{B}+r'_{B} > \frac{2N}{7}$. Hence, since clearly Bob has at most $\frac{4N}{7}$ votes in total, the correct leaf has the majority.
\end{proof}

\section{Coding Schemes with Small Alphabet and $O(n^2)$ Rounds}\label{sec:AlgLargeRounds}
In this section, we show how to implement the natural coding schemes presented as Algorithms \ref{alg:Simulator} and \ref{alg:SimulatorAdaptive} over a channel supporting only constant size symbols at the cost of increasing the number of rounds to $O(n^2)$. 

To emulate Algorithms \ref{alg:Simulator} and \ref{alg:SimulatorAdaptive} over a finite alphabet, we use Error Correcting Codes (ECC) and list decoding. In particular, on the transmission side, we encode each edge-set, which will have size at most $O(\frac{n}{\eps^2})$, using an ECC with relative distance $1-\eps/10$, alphabet size $O(\frac{1}{\eps})$, and code-length $O(\frac{n}{\eps^3})$ symbols and send this encoding in $O(\frac{n}{\eps^3})$ rounds. We call each such $O(\frac{n}{\eps^3})$ rounds related to transmission of one edge-set a \emph{block}. On the receiver side, we use a list decoder to produce a list of $O(\frac{1}{\eps})$ edge-sets such that, if the error-rate in the \emph{block} is at most $1-\eps/3$, then one of the edge-sets in the list in indeed equal to the transmitted edge-set. If the error-rate is greater than $1-\eps/3$, the list provided by the list decoder does not need to provide any guarantee. 

We use $O(\frac{n}{\eps^3})$ rounds for each block as because of list decoding, now each edge-set contains up to $O(\frac{n}{\eps^2})$ edges (compare to $O(\frac{n}{\eps})$ edges in Algorithms \ref{alg:Simulator} and \ref{alg:SimulatorAdaptive}). Adding edges corresponding to each of these edge-sets to the set of \emph{important} edges leads to the list decoding result:

\begin{lemma}
If the error rate is at most $1/2 - \eps$, then in both parties, the set of $O(1/\eps^2)$ leaves with the highest votes includes the correct leaf. 
\end{lemma}
\begin{proof}
This is because, with error-rate $1/2-\eps$, the adversary can corrupt at most $1/2-\eps/3$ blocks beyond corruption rate of $1-\eps/3$. Hence, we are now in a regime that at most $1/2-\eps/3$ fraction of edge-sets are corrupted, each corrupted possibly even completely, and for every other edge-set, the (list) decoding includes the correct transmitted edge-set. Hence, similar to the proof of \Cref{thm:nonadaptiveSim}, we get that the correct leaf gets a vote of at least $\Omega(\eps N)$. On the other hand, now each block might give a vote of at most $O({1}/{\eps})$ to different leaves and thus, the total weight is at most $O({N}/{\eps})$. Therefore, the correct leaf is within the top $O(1/\eps^2)$ voted leaves.
\end{proof}

We next explain a simple idea that allows us to extend this result to unique decoding; specifically non-adaptive unique decoding when error-rate is at most $1/4-\eps$, and adaptive unique decoding when error-rate is at most $2/7-\eps$.

The idea is to use a variation of Forney's technique for decoding concatenated codes~\cite{forney1966concatenated}. Recall that each received edge-set might lead to two things: (1) extending the common path, or (2) adding a vote to a leaf. While we keep the first part as above with list decoding, we make the voting weighted. In particular, in the receiver side, we take each edge-set leading to a leaf (when combined with local important edge set) as a vote for the related leaf but we weight this vote according to the hamming-distance to the received block. More precisely, if the edge-set has relative distance $\delta$ from the received block, the related leaf gets a vote of $\max\{1-2\delta, 0\}$. 

Using this weighting function, intuitively we have that if the adversary corrupts an edge-set to $\frac{1}{2}$ corruption rate, even though the edge-sets gets added to the set of important edges, in the weighting procedure this edge-set can only add at most weight $2\eps$ to any (incorrect) leaf. Hence, for instance if the adversary wants to add a unit of weight to the votes of an incorrect leaf, it has to almost completely corrupt the symbols of the related block.

\begin{lemma}\label{lem:14constant} If the error rate is at most $1/4 - \eps$, then in both parties, the leaf with the highest weighted votes is the correct leaf.  
\end{lemma}

\begin{proof}
We show that the correct leaf $u_{g}$ get strictly more weighted votes compared to any other leaf $u_b$. For this, we use a potential $\Phi=W(u_g)-W(u_b)$, that is, the total weight added to the good leaf minus that of the bad leaf, and we show this potential to be strictly positive. Let $P_{c}$ be the set of edges of the correct path.  Let $t$ be the time at which point the common path in $E_A \cup E_B\cup P_{c}$ ends in a leaf. First note that for each two consecutive blocks before time $t$ in which the corruption rate is at most $1-\eps/3$, the common path in $E_A \cup E_B\cup P_{c}$ gets extended by (at least) one edge (towards the correct leaf). Hence, at most $n$ such (``not completely corrupted") block pairs are spent on growing the common path in $E_A \cup E_B\cup P_{c}$ and also, $t$ happens after at most $N 2(1/2-2\eps)(1+\eps/3) <N (1-2\eps) < N-4n$ blocks. That is, $t$ happens at least $4n$ blocks before the end of the simulation. Before time $t$, we do not add any weight to $W(u_g)$ but each block corrupted with rate $x\geq 1/2$ changes $\Phi$ in the worst case as $1-2x\leq 0$. For each block after time $t$, each block corrupted with rate $x \in [0,1-\eps/3]$ changes $\Phi$ in the worst case by $1-2x$ and each block corrupted to rate $x\in (1-\eps/3, 1]$ changes $\Phi$ by at most $-1$. These two cases can be covered as a change of no worse than $1-2(1+\eps/3)x$. In total, since adversary's error rate is at most $(1/4-\eps)$, in total of before and after time $t$, we get that it can corrupt at most $1/2-2\eps$ fraction of the receptions of one party and thus, $\Phi \geq 1-2(1/2-2\eps)(1+\eps/3) \geq 3\eps >0$.
\end{proof}

For the $2/7-\eps$ adaptive algorithm, we first present a lemma about the total weight assigned to the leaves due to one edge-set reception. 

\begin{lemma}\label{lem:listweight}For each list decoded block that has corruption rate $\rho$, the summation of the weight given to all the codewords in the list is at most $|1-2\rho|+3\eps/5$.
\end{lemma}
\begin{proof} First we show that only at most $3$ codewords receive nonzero weight. The proof is by contradiction. Suppose that there are $4$ codewords $x_1$ to $x_4$ that each agree with the received string $x$ in at least $\eps/10$ fraction of the symbols. Furthermore, for each $x_i$, let $S_i$ be the set of coordinates where $x_i$ agrees with the received string $x$. Let $\ell$ be the length of string $x$. Note that $\forall i\neq j \in \{1, 2, 3, 4\}$, $|S_i|\geq \ell/2$ and $|S_i \cap S_j| \leq \eps/10$. Hence, with a simple inclusion-exclusion, we have $$\ell \geq |\cup_{i} S_i |\geq \sum_{i}|S_i| -\sum_{i<j} |S_i \cap S_j|\geq \frac{4\ell}{2} - \frac{6\eps}{10} >\ell,$$
which is a contradiction.

Having that at most $3$ codewords receive nonzero weight, we now conclude the proof. Let $x_1$ to $x_3$ be the codewords that receive nonzero weights and assume that $x_1$ is the closest codeword to $x$. We have $$\forall i> 1, \Delta(x, x_i) \geq \Delta(x_1, x_i)-\Delta(x, x_1)\geq 1-\eps/10 - 1/2 \geq 1/2-\eps/10.$$
Thus, the weight that $x_2$ or $x_3$ get is at most $2\eps/10$. On the other hand, $$\Delta(x_1, x) \geq \min\{\rho, 1-\eps/10-\rho\}= \frac{1-\eps/10}{2} - |\frac{1-\eps/10}{2} -\rho|.$$ Thus, the weight that $x_1$ receives is at most $\eps/10+|1-\eps/10-2\rho| \leq 2\eps/10+|1-2\rho|$. Summing up with the weight given to $x_2$ and $x_3$, we get that the total weight given to all codewords is at most $3\eps/5+|1-2\rho|$.  
\end{proof}

The algorithm is as in Algorithm \ref{alg:SimulatorAdaptive}, now enhanced with list decoding and weighted voting. In the end of the $\frac{6N}{7}$ rounds, a party decodes safely (and switched to only transmitting after that) if for the leaf $u$ that has the most votes, the following holds: $\Psi=\frac{W(u)+W_{\emptyset}-W(v)}{N} > 1/7$ where here, $W(v)$ is the weighted vote of the second leaf that has the most votes and $W_{\emptyset}$ is the total weight for decoded codewords that are inconsistent with the local important edge set (and thus mean an error). We call $\Psi$ the parties \emph{confidence}. 

The following lemma serves as completing the proof of \Cref{thm:27upperboundoverview}.

\begin{lemma} If the error rate is at most $2/7 - \eps$, then in the end, in both parties, the leaf with the highest weighted votes is the correct leaf.
\end{lemma}
\begin{proof}
We first show that at the end of the first $6N/7$ rounds, at least one party has confidence at least $1/7+\eps/2$. For this, we show the summation of the confidence of the two parties to be at least $2/7+\eps$. The reasoning is similar to the proof that of \Cref{lem:14constant}. Let $t$ be the time at which point the common path in $E_A \cup E_B\cup P_{c}$ ends in a leaf and note that for each two consecutive blocks before time $t$ in which the corruption rate is at most $1-\eps/3$, the common path in $E_A \cup E_B\cup P_{c}$ gets extended by (at least) one edge (towards the correct leaf). Hence, at most $n$ such (``not completely corrupted") block pairs are spent on growing the common path in $E_A \cup E_B\cup P_{c}$. Furthermore, $t$ happens after at most $2 \cdot N(2/7-\eps)(1+\eps/3) <N (4/7-\eps)$ blocks. Specifically, $t$ happens at least $N(2/7+\eps)$ blocks before the end of the first $6N/7$ blocks. On the other hand, for each corruption rate $x$ on one block, only the confidence of the party receiving it gets effected and in worst case it goes down by $\frac{1-2x(1+\eps/3)}{N}$. Since the adversary's total budget is $N(2/7-\eps)$, at the end of the first $6N/7$ rounds, the summation of the confidence of the two parties is at least $\frac{6N/7 - 2n - 2N(2/7-\eps)(1+\eps/3)}{N} \geq N(2/7+\eps)$.

Now we argue that if a party decodes at the end of $6N/7$ blocks because it has confidence at least $1/7$, then the decoding of this party is indeed correct. Proof is by contradiction. Using \Cref{lem:listweight}, each block with corruption rate $x$ can add a weight of at most $\max\{2x-1 +3\eps/5, 0\}$ to the set of bad leaves or those incorrect codeword edge-sets that do not form a common path with the local edge-sets, i.e., the weight given to $W_{\emptyset}$. Hence, knowing that the adversary has budget of $N(2/7-\eps)$, it can create a weight of at most $2N(2/7-\eps) -N/3(1-3\eps/5) < N/7 - \eps/5 < N/7$, which means that there can not be a confidence of more than $1/7$ on an incorrect leaf.

The above shows that at least one party decides after $6N/7$ blocks and if a party decides then, it has decided on the correct leaf. If a party does not decide, it listens for the next $N/7$ blocks where the other party is constantly transmitting. It is easy to see that in this case, this leaf that has the maximum vote in this listening party is the correct party.
\end{proof}

\section{Adaptivity in the Shared Randomness Setting}\label{sec:sharedrand}
In this section, we present our positive and negative results for adaptive protocols with access to shared randomness. 

In the shared randomness setting of Franklin et al. \cite{FGOS13} Alice and Bob have access to a shared random string that is not known to the adversary. As shown in Fanklin et al. \cite{FGOS13} the main advantage of this shared randomness is that Alice and Bob can use a larger communication alphabet for their protocol and agree on using only random subset of it which is independent from round to round. Since the randomness is not known to the adversary any attempt at corrupting a transmission results with good probability on an unused symbol. This makes most corruptions detectable to Alice and Bob and essentially results in corruptions being equivalent to erasures. Fanklin et al. call this way of encoding the desired transmissions into a larger alphabet to detect errors a \emph{blue-berry code}. It is well known for error correcting codes that erasures are in essence half as bad as corruption. Showing that the same holds for the use of tree codes in the algorithm of Braverman and Rao translates this factor two improvement to the interactive setting. For non-adaptive protocols this translates into a tolerable error rate of $1/2 - \eps$ instead of a $1/4 - \eps$ error rate. In what follows we show that allowing adaptive protocols in this setting allows to further rise the tolerable error rate to $2/3 - \eps$. We also show that no coding scheme can tolerate more than an $2/3$ error rate: 

\begin{theorem}\label{thm:xorshard23impossible}[The second part of \Cref{thm:23boundsoverview}]
In the shared randomness setting there is no (deterministic or randomized) adaptive protocol that robustly simulates the Exchange Protocol for an error rate of 2/3 with an $o(1)$ failure probability even when allowing computationally unbounded protocols that use an arbitrarily large number of rounds and an unbounded alphabet.
\end{theorem}

With the intuition that in the shared randomness setting corruptions are almost as bad as erasures we prove \Cref{thm:xorshard23impossible} directly for an strictly weaker adversary, which can only delete transmitted symbols instead of altering them to other symbols. Formally, when the adversary erases the transmission of a party, we say that the other party receives a special symbol $\bot$ which identifies the \emph{erasure} as such.

\begin{proof}
Suppose that there is a coding scheme with $N$ rounds that allows Alice and Bob exchange their input bits, while the adversary erases at most $2N/3$ of the transmissions. Consider each party in the special scenario in which this party receives a $\bot$ symbol whenever it listens. Let $x_A$ and $x_B$ be the (random variable of) the number of rounds where, respectively, Alice and Bob listen when they are in this special scenario. 

Suppose this number is usually small for one party, say Alice. That is, suppose we have $\Pr[x_A\leq 2N/3]\geq 1/3$. In this case the adversary gambles on the condition $x_A\leq 2N/3$ and simply erases Bob's transmission whenever Alice listens and Bob transmits. Furthermore, if both parties listen, the adversary delivers the symbol $\bot$ to both parties at no cost. This way, with probability at least $1/3$, Alice stays in the special scenario while the adversary adheres to its budget of at most $2N/3$ erasures. Since, in this case, Alice only receives the erasure symbol she cannot know Bob's input. Therefore, if the adversary chooses Bob's input to be a random bit the output of Alice will be wrong with probability at least $1/2$ leading to a total failure probability of at least $1/6$. 

If on the other hand $\Pr[x_A\leq 2N/3]\leq 1/3$ and $\Pr[x_B\leq 2N/3]\leq 1/3$, then even a union bound shows that $\Pr[(x_A\geq 2N/3) \wedge (x_A\geq 2N/3)]\geq 1/3$. In this case, the adversary tries to erase all transmissions of both sides. Indeed, if $x_A\geq 2N/3$ and $x_A\geq 2N/3$ this becomes possible because in this case there must be are at least $N/3$ rounds in which both parties are listening simultaneously. For these rounds, the adversary gets to deliver the erasure symbol to both sides at no cost. In this case, which happens with probability $1/3$ there remain at most $2N/3$ rounds in which one of the parties listens alone which the adversary can erase without running out of budget. With probability $\Pr[(x_A\geq 2N/3) \wedge (x_A\geq 2N/3)] > 1/3$ the adversary can thus prevent both parties from learning the other party's input. Choosing the inputs randomly results in at least one party being wrong in its decoding with probability $3/4$ which in total leads to a failure probability of at least $1/3 \cdot 3/4 > 1/6$. 
\end{proof}

Turning to positive results we first show how to solve the Exchange Problem robustly with an error rate of $2/3 - \eps$. While this is a much simpler setting it already provides valuable insights into the more complicated general case of robustly simulating arbitrary $n$-round protocols. For simplicity we stay with the adversary who can only erase symbols. Our general simulation result presented in \Cref{thm:protshard23} works of course for the regular adversary. 

\begin{lemma}\label{thm:xorshard23prot}
Suppose $\eps>0$. In the shared hidden randomness model with rate adaptation there is a protocol that solves the Exchange Problem in $O(1/\eps)$ rounds under an adversarial erasure rate of $2/3-\eps$.
\end{lemma}
\begin{proof}
The protocol consists of $3/\eps$ rounds grouped into three parts containing $1/\eps$ rounds each. In the first part Alice sends her input symbol in every round while Bob listens. In the second part Bob sends his input symbol in each round while Alice listens. In the last part each of the two parties sends its input symbol if and only if they have received the other parties input and not just erasures during the first two parts; otherwise a party listens during the last part. Note that the adversary has a budget of $3/\eps \cdot (2/3 - \eps) = 2\eps -1$ erasures which is exactly one erasure to little to erase all transmission during the first two parts. This results in at least one party learning the other party's input symbol. If both parties learn each others input within the first two rounds then the algorithm succeeded. If on the other hand one party, say without loss of generality Alice, only received erasures then Bobs received her input symbol at least once. This results in Bob sending his input symbol during the last part while Alice is listens. Bob therefore sends his input symbol a total of $2/\eps$ times while Alice listens. Not all these transmissions can be erased and Alice therefore also knows Bob's input symbol in the end. 
\end{proof}

\Cref{thm:xorshard23prot} and even more the structure of the robust Exchange Protocol presented in its proof already give useful insights into how to achieve a general robust simulation result. We combine these insights with the blue-berry code idea of \cite{FGOS13} to build what we call an \emph{adaptive exchange block}. An adaptive exchange block is a simple three round protocol which will serve as a crucial building block in \Cref{thm:protshard23}. An adaptive exchange block is designed to transmit one symbol $\sigma_A \in \Sigma$ from Alice to Bob and one symbol $\sigma_B \in \Sigma$ from Bob to Alice. It assumes a detection parameter $\delta < 1$ and works over an alphabet $\Sigma'$ of size $|\Sigma'|=\ceil{|\Sigma|/\delta}$ and works as follows: First, using the shared randomness, Alice and Bob use their shared randomness to agree for each of the three rounds on an independent random subset of of the larger communication alphabet $\Sigma'$ to be used. Then, in the first round of the adaptive exchange block Alice sends the agreed equivalent of $\sigma_A$ while Bob listens. In the second round Bob sends the equivalent of $\sigma_B$ while Alice listens. Both parties try to translate the received symbol back to the alphabet $\sigma$ and declare a (detected) corruption if this is not possible. In the last round of the adaptive exchange block a party sends the encoding of its $\sigma$-symbol if and only if a failure was detected; otherwise a party listens and tries again to decode any received symbol. 

The following two properties of the adaptive exchange block easily verified:

\begin{lemma}\label{lem:exchangeblock}
Regardless of how many transmissions an adversary corrupted during an adaptive exchange block, the probability that at least one of the parties decodes a wrong symbol is at most $3\delta$. 

In addition, if an adversary corrupted at most one transmission during an adaptive exchange block the with probability at least $1 - \delta$ both parties received their $\sigma$ symbols correctly. 
\end{lemma}
\begin{proof}
For a party to decode to a wrong symbol it must be the case that during one of the three rounds the adversary hit a meaningful symbol in the alphabet $\Sigma'$ during a corruption. Since $|\Sigma|/|\Sigma'|\leq \delta$ this happens at most with probability $\delta$ during any specific round and at most with probability $3\delta$ during any of the three rounds. To prove the the second statement we note that in order for a decoding error to happen the adversary must interfere during the first two rounds. With probability $1 - \delta$ such an interference is however detected leading to the corrupted party to resend in the third round. 
\end{proof}

Next we explain how to use the adaptive exchange block together with the ideas of \cite{FGOS13} to prove that adaptive protocols with shared randomness can tolerate any error rate below $2/3$:

\begin{theorem}\label{thm:protshard23}[The first part of \Cref{thm:23boundsoverview}]
Suppose $\Pi$ is an $n$-round protocol over a constant bit-size alphabet. For any $\eps >0$, there is an adaptive protocol $\Pi'$ that uses shared randomness and robustly simulates $\Pi$ for an error rate of $2/3 - \eps$ with a failure probability of $2^{-\Omega(n)}$. 
\end{theorem}
\begin{proof}[Proof Sketch]
The protocol $\Pi'$ consists of $N=\frac{6n}{\eps}$ rounds which are grouped into $\frac{2n}{\eps}$ adaptive exchange blocks. The adversary has an error budget of $\frac{4n}{\eps}-6n$ corruptions. Choosing the  with parameter $\delta$ in the exchange blocks to be at most $\eps/(6 \cdot 4 + 1)$ and using \Cref{lem:exchangeblock} together with a Chernoff bound gives that with probability $1 - 2^{-\Omega(n)}$ there are at most $n/4$ adaptive exchange blocks in which a wrong symbol is decoded. Similarly, the number of adaptive exchange blocks in which only one corruption occurred but not both parties received correctly is with probability $1 - 2^{-\Omega(n)}$ at most $n/3$. Lastly, the number of adaptive exchange blocks in which the adversary can force an erasure by corrupting two transmissions is at most $(\frac{4n}{\eps}-6n)/2$. We can therefore assume that regardless of the adversaries actions at most $n$ corruptions and at most $\frac{2n}{\eps} - 8/3 n$ erasures happen during any execution of $\Pi'$, at least with the required probability of $1 - 2^{-\Omega(n)}$. 

We can now apply the arguments given in \cite{FGOS13}. These arguments build on the result in \cite{BR11} and essentially show that detected corruptions or erasures can essentially be counted as half a corruption. Since almost all parts of the lengthy proofs in \cite{BR11} and \cite{FGOS13} stay the same, we restrict ourselves to a proof sketch. For this we first note that the result in \cite{BR11} continues to hold if instead of taking turns transmitting for $N$ rounds both Alice and Bob use $N/2$ rounds transmitting their next symbol simultaneously. The extensions described in \cite{FGOS13} furthermore show that this algorithm still performs a successful simulation if the number effectively corrupted simultaneous rounds is at least $n$ rounds less than half of all simultaneous rounds. Here the number of effectively corrupted simultaneous rounds is simply the number of simultaneous rounds with an (undetected) corruption plus half the number of simultaneous rounds suffering from an erasure (a detected corruption). Using one adaptive exchange block to simulate one simultaneous round leads to $\frac{2n}{\eps}$ simultaneous rounds and an effective number of corrupted rounds of $n/4 + (\frac{2n}{\eps} - 8/3 n)/2 < (\frac{2n}{\eps})/2 - n$. Putting everything together proves that with probability at least $1 - 2^{-\Omega(n)}$ the protocol $\Pi'$ successfully simulates the protocol $\Pi$, as claimed.

This way of constructing the adaptive protocol $\Pi'$ leads to an optimal linear number of rounds and a constant size alphabet but results in exponential time computations. We remark that using efficiently decodable tree codes, such as the ones described in the postscript to \cite{Schulman} on Schulman's webpage, one can also obtain a computationally efficient coding scheme at the cost of using a large $O(n)$-bit alphabet. Lastly, applying the same ideas as in \Cref{sec:AlgLargeRounds} also allows to translate this computationally efficient coding schemes with a large alphabet into one that uses a constant size alphabet but a quadratic number of rounds. 
\end{proof}

\appendix

\section{Adaptivity as a Natural Model for Shared-Medium Channels}\label{sec:appendix:naturalmodel}

In this section we briefly describe two natural interpretations of our adaptivity setting as modeling communication over a shared medium channel (e.g., a wireless channel). We first remark that one obvious parallel between wireless channels (or generally shared medium channels) and our model is that full-duplex communication, that is, sending and receiving at the same time, is not possible. 

Beyond this, one natural interpretation relates our model to a setting in which the signals used by these two parties are not much stronger than the average background noise level. In this setting having the background noise corrupt the signal a $\rho$ fraction of the time in an undetermined way is consistent with assuming that the variable noise level will be larger than the signal level at most this often. It is also consistent with the impossibility of distinguishing between the presence and absence of a signal which leads to undetermined signal decodings in the case that both parties listen. As always, the desire to avoid making possibly unrealistic assumptions about these corruptions naturally gives rise to think of any undetermined symbols as being worst case or equivalently as being determined by an adversary. 

In a second related interpretation of our model one can think of the adversary as an active malicious entity that jams the shared medium. In this case our assumptions naturally correspond to the jammer having an energy budget that allows to over-shout a sent signal at most a $\rho$ fraction of the time. In this setting it is also natural to assume that the energy required for sending a fake signal to both parties when no signal is present is much smaller than corrupting sent signals, and does as such not significantly reduce the jammer's energy budget.

\section{Further Results Supporting the Invariability Hypothesis}\label{sec:appendix:furtherresultssupportingthehypothesis}

In this section we mention several results that further support the Invariability Hypothesis.

We first remark that the lower bound in this paper are, in all settings and all properties, already as strong as required by the hypothesis. Our positive results, as summarized in \Cref{cor:weakhyp}, furthermore show that the invariability hypothesis holds if one allows either a large alphabet or a quadratic number of rounds. Both assumptions lead to the communication rate being $O(1/n)$. Next we list several results which show that the invariability hypothesis provably extends beyond these low rate settings \newblue{(note that some of the results from \cite{GH13} quoted below have been improved by the time of its publication~\cite{focs2014optimal})}:
\begin{enumerate}
  \item Subsequent results in \cite{GH13} show that:\\
	\emph{The IH is true for all eight settings when allowing randomized algorithms with exponentially small failure probability and a round blowup of $(\log^* n)^{O(\log^* n)}$.}

	\item Subsequent results in \cite{GH13} show that:\\
	\emph{The Invariability Hypothesis is true for all eight settings if one removes points 5. and 6., that is, when one can use randomized protocols that can generate private and public encryption keys which the adversary cannot break.}

	\item A subsequent result in \cite{GH13} gives an efficient randomized coding scheme for non-adaptive unique decoding which tolerates the optimal $1/4 - \eps$ error rate. This improves over the coding scheme of Brakerski-Kalai \cite{BK12} which requires an error rate below $1/16$. In different terms this result shows that:\\
	\emph{The Invariability Hypothesis is true for unique decoding if one weakens point 5. to allow efficient randomized algorithms with exponentially small failure probability.}

	\item The result of Braverman and Rao \cite{BR11} shows that:\\
	\emph{The Invariability Hypothesis is true for non-adaptive unique decoding if one removes point 4. (which requires protocols to be computationally efficient).}
	
	\item A subsequent result of the authors show that the $1/10 - \eps$ distance parameter of the tree code construction in \cite{Braverman12} can be improved to $1 - \eps$ which shows that:\\
	\emph{The Invariability Hypothesis is true for unique decoding if one weakens point 4. to allow deterministic sub-exponential time computations.}

	\item The improved tree code construction mentioned in point 5. can also be used together with the analysis of Franklin at al. \cite{FGOS13} to show that:\\
	\emph{The Invariability Hypothesis is true for non-adaptive unique decoding with shared randomness if one weakens point 4. to allow deterministic sub-exponential time computations.}
	
	\item The improved tree code construction mentioned in point 5. can also be used together with the ideas of \Cref{thm:protshard23} to show that:\\ 
	\emph{The Invariability Hypothesis is true for adaptive unique decoding with shared randomness if one weakens point 4. to allow deterministic sub-exponential time computations.}

\end{enumerate}

Lastly, we remark that the $1/6 - \eps$ tolerable error rate of the knowledge preserving coding schemes given by Chung et al. \cite{KnowledgePreservingIC} can be improved to the optimal $1/4 - \eps$ tolerable error rate of the unique decoding setting. The communication blowup for knowledge preserving protocols is however inherently super constant.

\section{Impossibility Results for List Decodable Interactive Coding Schemes}\label{sec:appendix:furtherlowerbounds} 

In this section we prove that list decodable interactive coding schemes are not possible beyond an error rate of $1/2$. This holds even if adaptivity or shared randomness is allowed (but not both). We remark that for both results we prove a lower bound for the $n$-bit Exchange Problem as list decoding with a constant list size becomes trivial for any protocol with only constantly many different inputs.

The intuitive reason why shared randomness does not help in non-adaptive protocols to go beyond an error rate of $1/2$ is because the adversary can completely corrupt the party that talks less: 

\begin{lemma}\label{lem:lowerboundLDsharedrandomness}
There is no (deterministic or randomized) list decodable non-adaptive protocol that robustly simulates the $n$-bit Exchange Protocol for an error rate of $1/2$ with an $o(1)$ failure probability and a list size of $2^{n-1}$ regardless of whether the protocols have access to shared randomness and regardless of whether computationally unbounded protocols that use an arbitrarily large number of rounds and an unbounded alphabet are allowed. 
\end{lemma}
\begin{proof}
We recall that non-adaptive protocols will for every round specify a sender and receiver in advance, that is, independent from the history of communication. We remark that the proof that follows continues to hold if these decisions are based on the shared randomness of the protocols.

The adversary's strategy is simple: It gives both Alice and Bob random inputs, then randomly picks one of them, and blocks all symbols sent by this party by replacing them with a fixed symbol from the communication alphabet. With probability at least $1/2$ the randomly chosen player speaks less than half the time and so the fraction of errors introduced is at most $1/2$. On the other hand no information about the blocked player's input is revealed to the other player and so other player can not narrow down the list of possibilities in any way. This means that even when allowed a list size of $2^{n-1}$ there is a probability of $1/2$ that the list does not include the randomly chosen input. This results in a failure probability of at least $1/4$. 
\end{proof}

A $1/2$ impossibility result also holds for list decodable coding schemes that are adaptive:

\begin{lemma}\label{lem:lowerboundLDadaptive}
There is no (deterministic or randomized) list decodable adaptive protocol that robustly simulates the $n$-bit Exchange Protocol for an error rate of $1/2$ with an $o(1)$ failure probability and a list size of $2^{n-1}$ even when allowing computationally unbounded protocols that use an arbitrarily large number of rounds and an unbounded alphabet. 
\end{lemma}

To show that adaptivity is not helpful, one could try to prove that the adversary can imitate one party completely without being detected by the other party. This, however, is not possible with an error rate of $1/2$ because both parties could in principle listen for more than half of the rounds if no error occurs and use these rounds to alert the other party if tempering is detected. Our proof of \Cref{lem:lowerboundLDadaptive} shows that this ``alert'' strategy cannot work. In fact, we argue that it is counterproductive for Alice to have such a hypothetical ``alert mode'' in which she sends more than half of the rounds. The reason is that the adversary could intentionally trigger this behavior while only corrupting less than half of the rounds (since Alice is sending alerts and not listening during most rounds). The adversary can furthermore do this regardless of the input of Bob which makes it impossible for Alice to decode Bob's input. This guarantees that Alice never sends for more than half of the rounds and the adversary can therefore simply corrupt all her transmissions. In this case Bob will not be able to learn anything about Alice's input. 

\begin{proof}
Suppose there exists a protocol that robustly simulates the $n$-bit Exchange Pro
tocol for an error rate of $1/2$ using $N$ rounds over an alphabet $\Sigma$. We 
consider pairs of the form $(x,\vec r)$ where $x \in \{0,1\}^n$ is an input to A
lice and $\vec r = (r_1,r_2,\ldots,r_N) \in \Sigma^N$ is a string over the chann
el alphabet with one symbol for each round. We now look at the following hypothetical communication between Alice and the adversary: Alice gets input $x$ and samples her private randomness. In each round she then decides to send or listen. If she listens in round $i$ she receives the symbol $r_i$. For every pair $(x,\vec r)$ we now define $p(x,\vec r)$ to be the probability, taken over the private randomness of Alice, that in this communication Alice sends at least $N/2$ rounds (and conversely listens for at most $N/2$ rounds). The adversaries strategy now makes the following case distinction: If there is one $(x,\vec r)$ for which $p(x,\vec r)>1/2$ then the adversary picks a random input for Bob, gives Alice input $x$ and during the protocol corrupts any symbol Alice listens to according to $\vec r$. By definition of $p(x,\vec r)$ there is a probability of at least $1/2$ that the adversary can do this without using more than $N/2$ corruptions. In such an execution Alice has furthermore no information on Bob's input and even when allowed a list size of $2^{n-1}$ has at most a probability of $1/2$ to include Bob's input into her decoding list. Therefore Alice will fail to list decode correctly with probability at least $1/4$. If on the other hand for every $(x,\vec r)$ it holds that $p(x,\vec r)<1/2$ then the adversary picks a random input for Alice and an arbitrary input for Bob and during the protocol corrupts any symbol Alice sends to a fixed symbol $\sigma \in \Sigma$. Furthermore in a round in which both Alice and Bob listens it chooses to deliver the same symbol $\sigma$ to Bob. By definition of $p(x,\vec r)$ there is a probability of at least $1/2$ that the adversary can do this without using more than $N/2$ corruptions. In such an execution Bob now has no information on Alice's input and even when allowed a list size of $2^{n-1}$ he therefore has at most a probability of $1/2$ to include Alice's input into his decoding list. This leads to a failure probability of $1/4$ as well. 
\end{proof}

\section{Impossibility Results for Solving the Exchange Problem Adaptively}\label{sec:Impossibilities}

In this appendix we provide the proofs we deferred in \Cref{sec:exchange}.

For completeness and to also cover randomized protocols we first reprove Claim 10 of \cite{BR11} which states that no non-adaptive uniquely decoding protocol can solve the exchange problem for an error rate of $1/4$:

\begin{lemma}[\textbf{Claim 10 of \cite{BR11}}] \label{lem:BRLB4}
There is no (deterministic or randomized) protocol that robustly simulates the Exchange Protocol for an error rate of 1/4 with an $o(1)$ failure probability even when allowing computationally unbounded protocols that use an arbitrarily large number of rounds and an unbounded alphabet.
\end{lemma}

\begin{proof}
Suppose that there is an algorithm $\mathcal{A}$ with no rate adaptation that solves the exchange problem under adversarial error-rate $1/4$, in $N$ rounds. We work simply with $1$-bit inputs. Let $S_{X,Y}$ denote the setting where Alice receives input $X$ and Bob gets input $Y$. For simplicity, we first ignore the rounds in which both parties listen; note that in those rounds the adversary can deliver arbitrary messages to each of the parties at no cost.

First we explain adversary's strategy for setting $S_{0,1}$: Consider the executions of $\mathcal{A}$ in setting $S_{0,0}$ with no interference from the adversary. Without loss of generality, we can assume that in this execution, with probability at least $1/2$, Alice listens in at most $1/2$ of the rounds. Noting the restriction that $\mathcal{A}$ has no rate-adaptivity, we get that in setting $S_{0,0}$, regardless of the adversary's interferences, with probability at least $1/2$, Alice listens alone in at most $1/2$ of the rounds. Adversary generates a dummy personality $\widetilde{Bob}_0$ by simulating algorithm $\mathcal{A}$ on Bob (and Alice) in setting $S_{0,0}$ with no interferences. This dummy personality is then used in setting $S_{0,1}$. Note that at the start, only Bob can distinguish $S_{0,1}$ from $S_{0,0}$. For the first $N/4$ times that Alice listens alone in $S_{0,1}$, the adversary connects Alice to the dummy $\widetilde{Bob}_0$, that is, Alice receives transmissions of $\widetilde{Bob}_0$. Thus, up to the point that Alice has listened alone for $N/4$ rounds, Alice receives inputs (with distribution) exactly as if she was in setting $S_{0,0}$ with real Bob and hence, she can not distinguish this setting from the setting $S_{0,0}$ with no adversarial interference. After this point, the adversary lets Alice talk freely with $Bob$ with no interference. 

We now explain adversary's strategy for setting $S_{0,0}$: The adversary generates another dummy personality $\widetilde{Bob}_1$ by simulating algorithm $\mathcal{A}$ on Bob (and Alice) in setting $S_{0,1}$ where the first $N/4$ listening-alone rounds of Alice were connected to $\widetilde{Bob}_0$. That is, exactly the situation that will happen to real Bob in setting $S_{0,1}$. For the first $N/4$ rounds of setting $S_{0,0}$ where Alice listens, the adversary does not interfere in the communications. After that, for the next $N/4$ rounds that Alice listens, the adversary delivers transmissions of dummy personality $\widetilde{Bob}_1$ to Alice.

To conclude the argument, the adversary give a random input $y\in\{0,1\}$ input to Bob and gambles on that Alice will be listening alone less than $1/2$ of the rounds. The adversary also uses the dummy personality $\widetilde{Bob}_i$ for $i=1-y$ and when Alice listens alone, the adversary connects Alice to the real Bob or this dummy personality according to the rules explained above. With probability at least $1/2$, Alice indeed listens in at most $N/2$ rounds. If this happens, due to the adversarial errors, the two settings look identical to Alice (more formally, she observes the same probability distributions for the inputs) and she can not distinguish them from each other. This means that algorithm $\mathcal{A}$ has a failure probability of at least $1/4$ (Alice can still guess $y$ but the guess would be incorrect with probability at at least $1/8$).

We finally explain the adversary's rules for treating the rounds where both parties listen: For setting $S_{0,0}$, if both Alice and the real Bob are listening, Alice is connected to $\widetilde{Bob}_1$ at no cost. Similarly, in setting $S_{0,1}$, if both Alice and real Bob are listening, then Alice is connected to $\widetilde{Bob}_0$ at no cost. To make sure that this definition does not have a loop, if for a round $r$, both parties are listening in both settings, then the adversary delivers a $0$ to Alice in both settings. Note that in using these rules, the behavior of the dummy personalities $\widetilde{Bob}_0$ and $\widetilde{Bob}_1$ are defined recursively based on the round number; for example, the simulation that generates the behavior of $\widetilde{Bob}_0$ in round $r$ might use the behavior of $\widetilde{Bob}_1$ in round $r-1$. Because of these rules, we get that in each round that Alice listens (at least until Alice has had $N/2$ listening-alone rounds), the messages that she receives have the same probability distribution in two settings and thus, the two settings look indistinguishable to Alice. If the execution is such that Alice listens alone less than $N/4$ rounds, which happens with probability at least $1/2$, the algorithm is bound to fail with probability at least $1/2$ in this case. This means algorithm $\mathcal{A}$ fails with probability at least $1/4$.
\end{proof}

Next we give the proof for \Cref{lem:xor13impossible} which shows that no adaptive protocol can robustly simulates the Exchange Protocol for an error rate of $2/7$:

\begin{proof}[Proof of \Cref{lem:xor13impossible}]
We first explain the adversaries strategy and then explain why this strategy makes at least one of the parties unable to know the input of the other party with probability greater than $1/2$. 

Suppose that there is an algorithm $\mathcal{A}$ that solves the exchange problem under adversarial error-rate $1/3$, in $N$ rounds. We work simply with $1$-bit inputs.  Let $S_{X,Y}$ denote the setting where Alice receives input $X$ and Bob gets input $Y$. For simplicity, we first ignore the rounds in which both parties listen; note that in those rounds the adversary can deliver arbitrary messages to each of the parties at no cost.

First we explain adversary's strategy for setting $S_{0,1}$: Consider setting $S_{0,0}$ and suppose that for the first $2N/3$ rounds in this, the adversary does not interfere. Without loss of generality, we can assume that with probability at least $1/2$, Alice listens alone in less than $N/3$ of these rounds. Adversary creates a dummy personality $\widetilde{Bob}_0$ and simulates it with Alice in this $S_{0,0}$ setting. In the first $2N/3$ rounds of setting $S_{0,1}$, whenever Alice listens alone, the adversary connects Alice to $\widetilde{Bob}_0$, that is, Alice receives the transmission of $\widetilde{Bob}_0$. With probability at least $1/2$ regarding the randomness of Alice and $\widetilde{Bob}_0$, Alice listens less than $N/3$ of the time which means that the adversary will have enough budget to completely fake Bob as $\widetilde{Bob}_0$ (from Alice's viewpoint). In that case, the two settings look identical to Alice for the first $2N/3$ rounds. During the last $N/3$ rounds of the execution in setting $S_{0,1}$, the adversary lets Alice and the real Bob talk without no interference. 

We now explain adversary's strategy for setting $S_{0,0}$: The adversary generates another dummy personality $\widetilde{Bob}_1$ by simulating Bob in setting $S_{0,1}$ where alone-listenings of Alice in the first $2N/3$ rounds are connected to $\widetilde{Bob}_0$. In setting $S_{0,0}$, the adversary lets Alice and Bob talk freely during the first $2N/3$ rounds but for the last $N/3$ rounds, whenever Alice listens (even if not alone), the adversary connects her to $\widetilde{Bob}_1$. 

To conclude, the adversary give a random input $y\in\{0,1\}$ input to Bob and gambles on that Alice listens alone less than $N/3$ rounds of the first $2N/3$ rounds. The adversary also uses the dummy personality $\widetilde{Bob}_i$ for $i=1-y$ and when Alice listens alone, the adversary connects Alice to the real Bob or this dummy personality according to the rules explained above. We know that with probability at least $1/2$, Alice listens alone less than $N/3$ rounds of the first $2N/3$ rounds. If that happens, with at most $N/3$ errors, the adversary can follow the strategy explained above. Therefore, with probability $1/2$, Alice can not know Bob's input and thus will fail with probability at least $1/4$.

Regarding the rounds where both parties are listening, the rule is similar to \Cref{lem:BRLB4} but a little bit more subtle because of the rate adaptivity of algorithm $\mathcal{A}$: We need to declare what are the receptions when in the first $2N/3$ rounds of setting $S_{0,0}$, both Alice and Bob are listening. However, at that point, it's not clear whether we will make an indistinguishability argument for Alice or for Bob, which since it is affected by which one of them listens alone less, it can also depend on the receptions during the all-listening rounds of the first $2N/3$ rounds. The simple remedy is to create analogous situations for Alice and Bob so that we can determine the base of indistinguishability later at the end of $2N/3$ rounds. Adversary generates dummy personalities $\widetilde{Alice}_0$ and $\widetilde{Alice}_1$, respectively in settings $S_{0,0}$ and $S_{1,0}$, similar to those of Bob. In each all-listening round of the first $2N/3$ rounds, adversary makes Alice receive the transmission of $\widetilde{Bob}_1$ and Bob receive the transmission of $\widetilde{Alice}_1$. With these connections, whoever is more likely to listen alone less in the first $2N/3$ rounds ---which we assumed to be Alice without loss of generality in the above discussions--- with constant probability will receive messages with the exact same probability distributions, in each round in the two different settings. Thus she will not be able to distinguish the two settings.
\end{proof}

\end{document}